\documentclass{fundam}
\usepackage{url} 
\usepackage[ruled,lined]{algorithm2e}
\usepackage{graphicx}

\usepackage{amssymb}
\usepackage{amsmath}
\usepackage{times}
\usepackage{enumitem}
\usepackage{mymacros}

\def\cluster#1{\ensuremath{[#1]_c}}

\usepackage[normalem]{ulem}
 \newcommand{\stkout}[1]{\ifmmode\text{\sout{\ensuremath{#1}}}\else\sout{#1}\fi}

\begin{document}
\setcounter{page}{1}

\title{Free-Choice Nets With Home Clusters Are Lucent\footnote{This paper was submitted  to Fundamenta Informaticae on 9-8-2020 and accepted on 2-6-2021.}}

\author{Wil M.P. van der Aalst \\
Process and Data Science (PADS)\\ 
RWTH Aachen University, Germany\\
wvdaalst{@}pads.rwth-aachen.de}

\runninghead{W.M.P. van der Aalst}{Free-Choice Nets With Home Clusters Are Lucent}
\maketitle

\begin{abstract}
A marked Petri net is \emph{lucent} if there are no two different reachable markings enabling the same set of transitions, i.e.,
states are fully characterized by the transitions they enable.
Characterizing the class of systems that are lucent is a foundational and also challenging question.
However, little research has been done on the topic.
In this paper, it is shown that all \emph{free-choice nets having a home cluster} are lucent. 
These nets have a so-called home marking such that it is always possible to reach this marking again.
Such a home marking can serve as a regeneration point or 
as an end-point.
The result is highly relevant because 
in many applications, we want the system to be lucent and 
many ``well-behaved'' process models fall into the class identified in this paper.
Unlike previous work, we do not require the marked Petri net to be live and strongly-connected.
Most of the analysis techniques for free-choice nets are tailored towards well-formed nets.
The approach presented in this paper provides a novel perspective
enabling new analysis techniques for free-choice nets that do not need to be well-formed.
Therefore, we can also model systems and processes that are terminating and/or have an initialization phase.
\end{abstract}


\section{Introduction}
\label{sec:intro}

Petri nets can be used to model systems and processes.
Many properties have been defined for Petri nets
that describe desirable characteristics of the modeled system or process \cite{lopn1,structure-theory-ToPNoC-advanced-course2010,murata}.
Examples include 
deadlock-freeness (the system is always able to perform an action),
liveness (actions cannot get disabled permanently),
boundedness (the number is states is finite),
safeness (objects cannot be at the same location at the same time),
soundness (a case can always terminate properly) \cite{soundness-FACS}, etc.
In this paper, we investigate another foundational property: \emph{lucency}. 
A system is lucent if it does not have different reachable states that enable the same actions, i.e., the set of enabled actions
uniquely characterizes the state of the system \cite{lucent-PN2018}.
Think of an information system that has a user interface showing what the user can do. In this example, lucency implies that the offered actions fully determine the internal state
and the system will behave consistently from the user's viewpoint.
If the information system would not be lucent, the user could encounter situations where the set of offered actions is the same, but the behavior is very different.
Another example is the worklist of a workflow management system that shows the workitems that can or should be executed. 
Lucency implies that the state of a case can be derived based on the workitems offered for it.

In a Petri net setting, lucency can be defined as follows.
\emph{A marked Petri net is lucent if there are no two different reachable markings enabling the same set of transitions, i.e.,
markings are fully characterized by the transitions they enable.}
\begin{figure}[thb!]
\centering
\includegraphics[width=9.0cm]{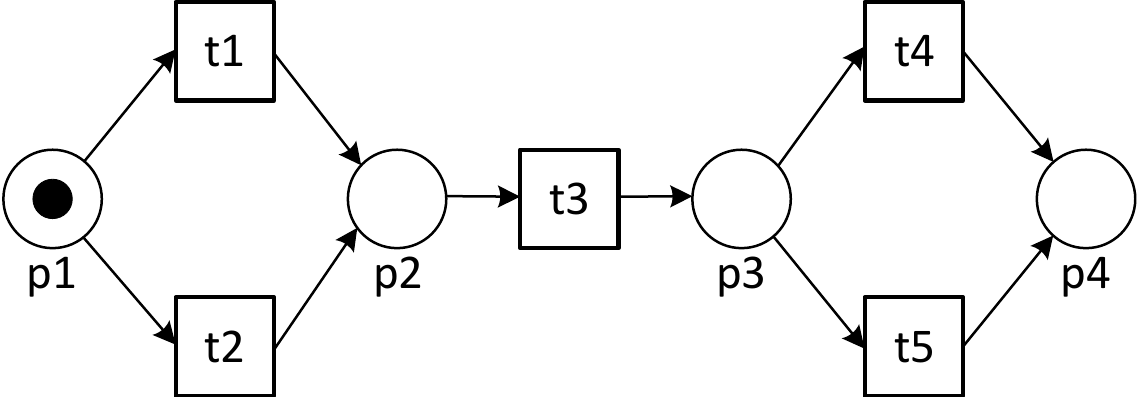}
\caption{$(N_1,M_1)$ is a free-choice net that is lucent, has a home cluster, but is not perpetual.}
\label{fig-intro-fc}
\end{figure}

Figure~\ref{fig-intro-fc} shows a marked Petri net that is lucent.
Each of the four reachable markings has a different set of enabled transitions.
Figure~\ref{fig-intro-nfc} shows a marked Petri net that is not lucent. Initially, one of the transitions $t1$ or $t2$ can occur, leading to two different states (the markings $[p2,p5]$ and $[p2,p6]$) that cannot be distinguished.
Only transition $t3$ is enabled, but the internal state matters.
$t1$ is always followed by $t4$ and $t2$ is always followed by $t5$.
\begin{figure}[thb!]
\centering
\includegraphics[width=9.0cm]{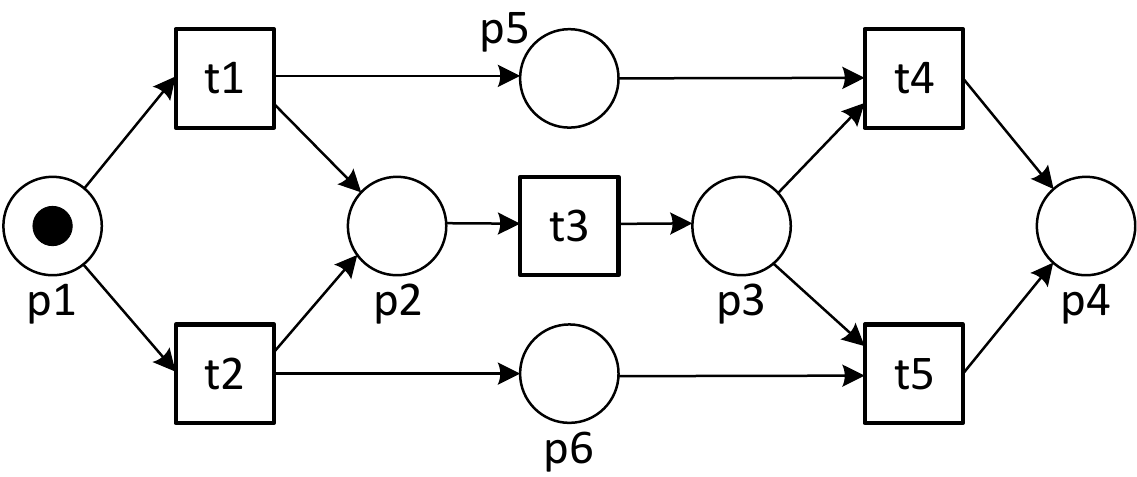}
\caption{$(N_2,M_2)$ is a non-free-choice net that is not lucent because the markings $[p2,p5]$ and $[p2,p6]$ enable the same set of transitions (just $t3$), thereby hiding the internal state.}
\label{fig-intro-nfc}
\end{figure}

Although we focus on Petri nets, lucency is a general notion that is independent of the modeling language used.
Even though lucency is an easy to define and foundational property, 
it was not investigated until recently \cite{lucent-PN2018,lucent-translucent-fi-2019}. As described in \cite{lucent-translucent-fi-2019},
lucent process models are easier to discover from event data.
When the underlying process has states that are different, 
but that enable the same set of activities, 
then it is obviously not easy to learn these ``hidden'' states.
Commercial process mining systems mostly use the so-called Directly-Follows Graph (DFG) as a process model. Here the ``state'' is considered to be the last activity executed.
DFGs have problems dealing with concurrent processes and tend to produce imprecise and ``Spaghetti-like'' models because of that. More advanced process discovery techniques are able to discover concurrent process models \cite{process-mining-book-2016}, but need to ``guess'' the state of the process after each event. When using, for example, region theory, the state is often assumed to be the prefix of activities (or the multiset of activities already executed), leading to overfitting and incompleteness problems (one needs to see all possible prefixes).
For lucent process models, this problem is slightly easier because the state is fully determined by the set of enabled activities. See \cite{lucent-translucent-fi-2019} for more details about the discovery of lucent process models using translucent event logs.

Given the examples in Figures~\ref{fig-intro-fc} and \ref{fig-intro-nfc}, there seems a natural connection between the well-known free-choice property \cite{deselesparza} and lucency. 
In a free-choice net, choice and synchronization can be separated.
However, as illustrated by Figure~\ref{fig-locally-safe-not-perpetual}, it is not enough to require that the net is free-choice. $(N_3,M_3)$ shown in Figure~\ref{fig-locally-safe-not-perpetual} is free-choice. It is actually a marked graph since there are no choices (i.e., places with multiple output arcs). The model in Figure~\ref{fig-locally-safe-not-perpetual} satisfies most of the (often considered desirable) properties defined for Petri nets.
$(N_3,M_3)$ is deadlock-free, live, bounded, safe, well-formed, free-choice, all markings are home markings, etc.
However, surprisingly $(N_3,M_3)$ is not lucent because the two reachable markings $[p1,p3,p6]$ and $[p1,p4,p6]$ enable the same set of transitions ($t1$ and $t4$).
This example shows that lucency does not coincide with any (or a combination) of the properties normally considered.
\begin{figure}[thb!]
\centering
\includegraphics[width=8.0cm]{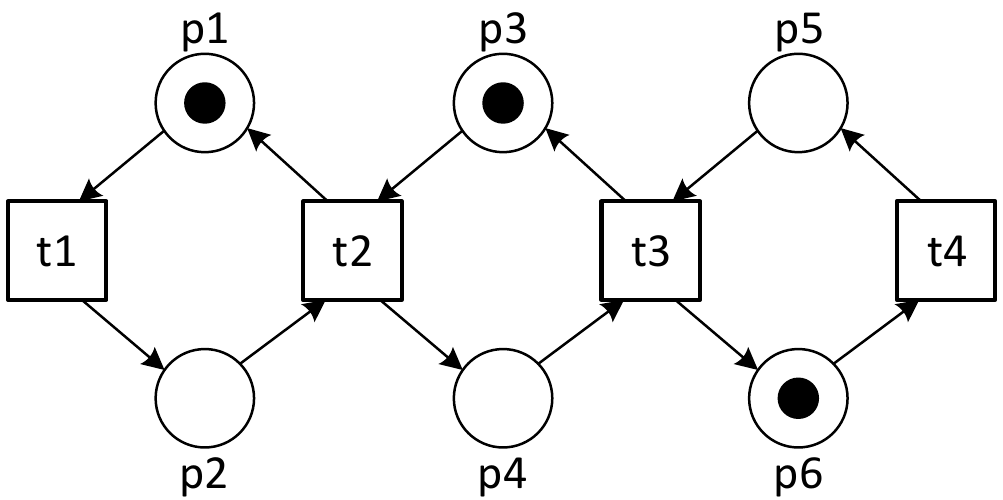}
\caption{$(N_3,M_3)$ is a marked graph that is not lucent because the markings $[p1,p3,p6]$ and $[p1,p4,p6]$ enable the same set of transitions ($t1$ and $t4$), thereby hiding the internal state.}
\label{fig-locally-safe-not-perpetual}
\end{figure}

The notion of lucency was first introduced in \cite{lucent-PN2018}. The paper uses the example shown in Figure~\ref{fig-locally-safe-not-perpetual}
to demonstrate that even nets that are free-choice, live, and safe may not be lucent. Therefore, an additional requirement was added.
In \cite{lucent-PN2018}, the class of \emph{perpetual} nets is introduced in an attempt to relate well-known Petri net properties to lucency. Perpetual free-choice nets are free-choice Petri nets that are live and bounded and have a home cluster, i.e., there is a cluster such that from any reachable state,
there is a reachable state marking the places of this cluster. 
Such a home cluster in a perpetual net 
serves as a ``regeneration point'' of the process, e.g., to start a new process instance (case, job, cycle, etc.).
Any perpetual marked free-choice net is lucent.
However, there are many lucent systems that are
not perpetual because they are terminating or have an initialization phase (and are therefore not live).

This paper extends the work presented in \cite{lucent-PN2018}
which focused exclusively on perpetual marked free-choice nets.
For example, $(N_1,M_1)$ in Figure~\ref{fig-intro-fc} 
is not perpetual. Actually, most of the work done on free-choice nets is limited to well-formed nets, i.e., nets that have a marking that is live and bounded. This is a structural property allowing for many interesting and advanced forms of analysis and reasoning \cite{bestfcn,structure-theory-ToPNoC-advanced-course2010,deselesparza}.
Such nets are automatically strongly-connected and do not have source and sink places to model 
the start and the end of the process.

However, in many applications, such nets are not suitable.
For example, it is impossible to model systems and processes that can terminate.
In some cases, one can apply a trick and ``short-circuit'' 
the actual net to make it well-formed (see, for example, the analysis of soundness for workflow nets \cite{aaljcsc}).
However, this distracts from the essence of the property being analyzed. 
This paper proves this point by showing that liveness is \emph{irrelevant} for ensuring lucency.
For example, the Petri net in Figure~\ref{fig-intro-fc} is lucent, but not well-formed.

In this paper, we show that all \emph{free-choice nets having a home cluster are lucent}. This significantly extends the 
class perpetual marked free-choice nets and also includes 
non-well-formed nets such as $(N_1,M_1)$ in Figure~\ref{fig-intro-fc}.

To do this, we provide a direct proof 
that is \emph{not} building on
the traditional stack of results for well-formed free-choice nets.
In \cite{lucent-PN2018}, we need to use the 
coverability theorem and the blocking marking theorem.
Moreover, the proof in \cite{lucent-PN2018} turned out to be incomplete and the repaired proof is even more involved.
\emph{The approach used to prove the correctness of the main result provides a novel perspective
enabling new analysis techniques for free-choice nets that do not need to be well-formed.} 
Novel concepts like ``expediting transitions'', ``rooted disentangled paths'', and ``conflict-pairs'' can  be used to prove many other properties free-choice nets having a home cluster.
This paper also relates the novel concepts and techniques presented in this paper to results based on 
short-circuiting nets that are non-live and not strongly-connected (Section~\ref{sec:perpmarknets}).
This relation is used to show that we can check whether there is a home cluster in polynomial time for free-choice nets (whether they are live and strongly-connected or not).

The remainder is organized as follows.
Section~\ref{sec:rw} briefly discusses related work.
Section~\ref{sec:prelim} introduces Petri nets and some of the basic notations.
Lucent Petri nets are defined in Section~\ref{sec:lucency}.
In Section~\ref{sec:arelucent} we show that free-choice nets having a home cluster are indeed lucent.
To do this, we introduce new notions such as (rooted)
disentangled paths and conflict-pairs.
Section~\ref{sec:perpmarknets} relates the work to perpetual marked free-choice nets and our earlier paper \cite{lucent-PN2018}.
Section~\ref{sec:concl} concludes the paper.

\section{Related Work}
\label{sec:rw}

This paper extends the work presented in \cite{lucent-PN2018}.
There are no other papers on the analysis of lucency (which is surprising).
Hence, we can only point to more indirectly related work.

For more information about Petri nets, we refer to \cite{mbp-aal-stahl-2011,murata,reisig-book-2013,lopn1,lopn2}.
Within the field of Petri nets ``structure theory'' plays an important role \cite{bestfcn,structure-theory-ToPNoC-advanced-course2010,deselesparza}.
Free-choice nets are well studied \cite{bdefcn,structure-theory-ToPNoC-advanced-course2010,Esparza98TCS,thiavoss}.
The definite book on the structure theory of free-choice nets is \cite{deselesparza}.
Also, see \cite{structure-theory-ToPNoC-advanced-course2010} for pointers to literature.
Therefore, it is surprising that the question of whether markings are uniquely identified by the set of enabled transitions (i.e., lucency)
has not been explored in literature.
Lucency is unrelated to the so-called ``frozen tokens'' \cite{frozen-tokens-wehler2010}.
A Petri net has a frozen token if there exists an infinite occurrence sequence never using the token.
It is possible to construct live and bounded free-choice nets that are lucent while having frozen tokens.
Conversely, there are live and bounded free-choice nets that do not have frozen tokens and are not lucent.

The results presented in this paper are also related to the  \emph{blocking theorem}
\cite{blocking-theorem-gaujala2003,blocking-theorem-wehler2010}.
Blocking markings are reachable markings that enable transitions from only a single cluster. Removing the cluster yields a dead marking.
The blocking theorem states that in a bounded and live free-choice net each cluster has a unique blocking marking.
Lucency is broader than blocking markings since multiple clusters and concurrent transitions are considered.
Actually, lucency can be seen as a generalization of unique blocking markings.
Moreover, \cite{blocking-theorem-gaujala2003,blocking-theorem-wehler2010} only consider live Petri nets.

In \cite{reduction-wvda-PN2021}, we propose a framework based on sequences of $t$-induced T-nets and 
$p$-induced P-nets to convert free-choice nets into T-nets and P-nets while preserving properties such as
well-formedness, liveness, lucency, pc-safety, and perpetuality. The framework allows for systematic
proofs that ``peel off'' non-trivial parts while retaining the essence of the problem
(e.g., lifting properties from T-nets and P-nets to free-choice nets).

A major difference between the work reported in this paper and 
the extensive body of knowledge just mentioned is that we do \emph{not} require the Petri net to be well-formed. Liveness assumes that the system is cyclic and actions are always still possible in the future. This does not align well with the standard ``case notion'' used
in Business Process Management (BPM), Workflow Management (WFM), and Process Mining (PM) \cite{aaljcsc,process-mining-book-2016,soundness-FACS}. 
Process instances have a clear start and end.
For example, process discovery algorithms from the field of PM all generate process models close to the workflow nets.
The languages used for BPM and WFM, e.g., BPMN and UML Activity Diagrams, are very different from well-formed Petri nets and closer to workflow nets.
The work presented in this paper supports both views.
The process models may be well-formed or not. 
Therefore, we significantly generalize over the work presented in \cite{lucent-PN2018} and also present results that could be used for other questions.

\section{Preliminaries}
\label{sec:prelim}

This section introduces concepts related to Petri nets and some basic notations.

\subsection{Multisets, Sequences, and Functions}
\label{sec:basics}

$\bag(A)$ is the set of all \emph{multisets} over some set $A$.
For some multiset $b\in \bag(A)$, $b(a)$ denotes the number of times element $a\in A$ appears in $b$.
Some examples: $b_1 = [~]$, $b_2 = [x,x,y]$, $b_3 = [x,y,z]$, $b_4 = [x,x,y,x,y,z]$, and $b_5 = [x^3,y^2,z]$
are multisets over $A=\{x,y,z\}$.
$b_1$ is the empty multiset, $b_2$ and $b_3$ both consist of three elements, and
$b_4 = b_5$, i.e., the ordering of elements is irrelevant and a more compact notation may be used for repeating elements.
The standard set operators can be extended to multisets, e.g., $x\in b_2$, $b_2 \bplus b_3 = b_4$, $b_5 \setminus b_2 = b_3$, $|b_5|=6$, etc.
$\{a \in b\}$ denotes the set with all elements $a$ for which $b(a) \geq 1$.
$b(X) = \sum_{a \in X}\ b(x)$ is the number of elements in $b$ belonging to set $X$, e.g., $b_5(\{x,y\}) = 3+2=5$.
$b \leq b'$ if $b(a) \leq b'(a)$ for all $a \in A$. Hence, $b_3 \leq b_4$ and $b_2 \not\leq b_3$ (because $b_2$ has two $x$'s).
$b < b'$ if $b \leq b'$ and $b \neq b'$.
Hence, $b_3 < b_4$ and $b_4 \not< b_5$ (because $b_4 = b_5$).

$\sigma = \langle a_1,a_2, \ldots, a_n\rangle \in X^*$ denotes a \emph{sequence} over $X$ of length $\card{\sigma} = n$.
$\sigma_i = a_i$ for $1 \leq i \leq \card{\sigma}$.
$\langle~\rangle$ is the empty sequence. 
$\sigma_1 \cdot \sigma_2$ is the concatenation of two sentences, e.g., $\langle x,x,y\rangle \cdot \langle x,z\rangle = \langle x,x,y,x,z\rangle$.
The notation $[a \in \sigma]$ can be used to convert a sequence into a multiset. $[a \in \langle x,x,y,x,z\rangle] = [x^3,y^2,z]$.

\subsection{Petri Nets}
\label{sec:petrinets}

Figures~\ref{fig-intro-fc}, \ref{fig-intro-nfc},  and \ref{fig-locally-safe-not-perpetual} already showed examples of marked Petri nets. To reason about such processes and to formalize lucency, we now provide the basic formalizations \cite{mbp-aal-stahl-2011,murata,reisig-book-2013,lopn1,lopn2}.

\begin{definition}[Petri Net]\label{def:pn}
A \emph{Petri net} is a tuple $N=(P,T,F)$ with $P$ the non-empty set of places, $T$ the non-empty set of transitions such that
$P \cap T = \emptyset$, and $F\subseteq (P \times T) \cup (T \times P)$ the flow relation such that the graph $(P \cup T, F)$ is (weakly) connected. 
\end{definition}

Figure~\ref{fig-intro-fc} has 
four places ($p1,p2,p3,p4$),
five transitions ($t1,t2,t3,t4,t5$),
and ten arcs. 
The initial marking contains just one token located in place $p1$.

\begin{definition}[Marking]\label{def:mrk}
Let $N=(P,T,F)$ be a Petri net.
A \emph{marking} $M$ is a multiset of places, i.e., $M \in \bag(P)$.
$(N,M)$ is a marked net.
\end{definition}

A Petri net $N=(P,T,F)$ defines a directed graph with nodes $P\cup T$ and edges $F$.
For any $x\in P\cup T$, $\pre{x} = \{y\mid (y,x)\in F\}$ denotes the set of input nodes and
$\post{x} = \{y\mid (x,y)\in F\}$ denotes the set of output nodes.
The notation can be generalized to sets: $\pre{X}=\{y\mid \exists_{x\in X} \ (y,x)\in F\}$ and $\post{X} = \{y\mid \exists_{x\in X} \ (x,y)\in F\}$
for any $X \subseteq P\cup T$. 

A transition $t \in T$ is \emph{enabled} in marking $M$ of net $N$, denoted as $(N,M)[t\rangle$, if each of its input places ${\pre{t}}$ contains at least one token.
$\mi{en}(N,M) = \{ t \in T \mid (N,M)[t\rangle \}$ is the set of enabled transitions.

An enabled transition $t$ may \emph{fire}, i.e., one token is removed from each of the input places ${\pre{t}}$ and
one token is produced for each of the output places ${\post{t}}$.
Formally: $M' = (M\bmin {\pre{t}})\bplus {\post{t}}$ is the marking resulting from firing enabled transition $t$ in marking $M$ of Petri net $N$.
$(N,M)[t\rangle (N,M')$ denotes that $t$ is enabled in $M$ and firing $t$ results in marking $M'$.

Let $\sigma = \langle t_1,t_2, \ldots, t_n \rangle \in T^*$ be a sequence of transitions.
$(N,M)[\sigma\rangle (N,M')$ denotes that there is a set of markings $M_1, M_2, \ldots, M_{n+1}$ ($n \geq 0$)
such that $M_1 = M$, $M_{n+1} = M'$, and $(N,M_i)[t_i\rangle (N,M_{i+1})$ for $1 \leq i \leq n$.
A marking $M'$ is \emph{reachable} from $M$ if there exists a \emph{firing sequence} $\sigma$ such that $(N,M)[\sigma\rangle (N,M')$.
$R(N,M) = \{ M' \in \bag(P) \mid \exists_{\sigma \in T^*} \ (N,M)[\sigma\rangle (N,M') \}$ is the set of all reachable markings.
$(N,M)[\sigma\rangle$ denotes that the sequence $\sigma$ is enabled when starting in marking $M$ (without specifying the resulting marking).

For the marked net in Figure~\ref{fig-intro-nfc}: $R(N_2,M_2)= \{ [p1],[p2,p5],[p2,p6],[p3,p5],[p3,p6],[p4]\}$.
Note that $\mi{en}(N_2,[p2,p5]) = \mi{en}(N_2,[p2,p6]) = \{t3\}$.

\subsection{Liveness, Boundedness, and Home Markings}
\label{sec:livenessetc}

Next, we define some of the standard behavioral properties for Petri nets: liveness, boundedness, and the presence of home markings.

\begin{definition}[Live, Bounded, Safe, Dead, Deadlock-free, Well-Formed]\label{def:lb}
A marked net $(N,M)$ is \emph{live} if for every reachable marking $M' \in R(N,M)$ and for every transition $t\in T$ there exists a marking $M'' \in R(N,M')$ that enables $t$.
A marked net $(N,M)$ is $k$-bounded if for every reachable marking $M' \in R(N,M)$ and every $p \in P$: $M'(p) \leq k$.
A marked net $(N,M)$ is \emph{bounded} if there exists a $k$ such that $(N,M)$ is $k$-bounded.
A 1-bounded marked net is called \emph{safe}.
A place $p\in P$ is \emph{dead} in $(N,M)$ when it can never be marked (no reachable marking marks $p$).
A transition $t\in T$ is \emph{dead} in $(N,M)$ when it can never be enabled (no reachable marking enables $t$).
A marked net $(N,M)$ is \emph{deadlock-free} if each reachable marking enables at least one transition.
A Petri net $N$ is \emph{structurally bounded} if $(N,M)$ is bounded for any marking $M$.
A Petri net $N$ is \emph{structurally live} if there exists a marking $M$ such that $(N,M)$ is live.
A Petri net $N$ is \emph{well-formed} if there exists a marking $M$ such that $(N,M)$ is live and bounded.
\end{definition}

\begin{definition}[Home Marking]\label{def:hm}
Let $(N,M)$ be a marked net.
A marking $M_H$ is a \emph{home marking} if for every reachable marking $M' \in R(N,M)$: $M_H \in R(N,M')$.
\end{definition}

Note that home markings do not imply liveness or boundedness, i.e., a Petri net may be non-well-formed and still have home markings. $(N_1,M_1)$ in Figure~\ref{fig-intro-fc} is not live and has one home marking $[p4]$. $(N_3,M_3)$ in Figure~\ref{fig-locally-safe-not-perpetual} is live and all of its reachable markings are home markings.

\subsection{Clusters}
\label{sec:clusters}

Clusters play a major role in this paper.
A cluster is a maximal set of connected nodes, only considering
arcs connecting places to transitions.

\begin{definition}[Cluster]\label{def:clust}
Let $N=(P,T,F)$ be a Petri net and $x \in P \cup T$.
The \emph{cluster} of node $x$, denoted $\cluster{x}$ is the smallest set such that
(1) $x \in \cluster{x}$,
(2) if $p \in \cluster{x} \cap P$, then ${\post{p}} \subseteq \cluster{x}$, and
(3) if $t \in \cluster{x} \cap T$, then ${\pre{t}} \subseteq \cluster{x}$.
$\cluster{N}= \{ \cluster{x} \mid x \in P \cup T\}$ is the set of clusters of $N$.
\end{definition}

Note that $\cluster{N}$ partitions the nodes in $N$.
The Petri net in Figure~\ref{fig-intro-fc} has four clusters:
$C_1 = \{p1,t1,t2\}$, 
$C_2 = \{p2,t3\}$, 
$C_3 = \{p3,t4,t5\}$, and
$C_4 = \{p4\}$.
The Petri net in Figure~\ref{fig-locally-safe-not-perpetual} 
also has four clusters:
$C_1 = \{p1,t1\}$, 
$C_2 = \{p2,p3,t2\}$, 
$C_3 = \{p4,p5,t3\}$, and
$C_4 = \{p6,t4\}$.

\begin{definition}[Cluster Notations]\label{def:clustnot}
Let $N=(P,T,F)$ be a Petri net and $C \in \cluster{N}$ a cluster.
$\mi{Pl}(C) = P \cap C$ are the places in $C$, $\mi{Tr}(C) = T \cap C$ are the transitions in $C$, and $\mi{Mrk}(C) = [p \in \mi{Pl}(C)]$ is the smallest marking fully enabling the cluster.
\end{definition}

\subsection{Structural Properties}
\label{sec:structprop}

As defined before, 
we require Petri nets to be \emph{weakly} connected.
$N$ is \emph{strongly} connected if the graph $(P \cup T,F)$ is strongly-connected, i.e., for any two nodes $x$ and $y$ there is a path leading from $x$ to $y$.

Various subclasses of Petri nets have been defined
based on the network structures they allow.
State machines, also called P-nets, do not allow for transitions with multiple input or output places.
Marked graphs, also called T-nets, do not allow for places with multiple input or output transitions.
In this paper, we focus on \emph{free-choice} nets that are \emph{proper}.

\begin{definition}[Free-choice Net]\label{def:fcne}
Let $N=(P,T,F)$ be a Petri net.
$N$ is \emph{free-choice net} if for any $t_1,t_2 \in T$: ${\pre{t_1}} = {\pre{t_2}}$ or ${\pre{t_1}} \cap {\pre{t_2}} = \emptyset$.
\end{definition}

In free-choice nets, choice and synchronization can be separated.
$(N_2,M_2)$ in Figure~\ref{fig-intro-nfc} is not a free-choice net, because the choice between $t4$ and $t5$ is controlled by the places $p5$ and $p6$.

\begin{definition}[Proper Petri Net]\label{def:proppn}
A Petri net $N=(P,T,F)$ is \emph{proper} if all transitions have input and output places, i.e., for all $t \in T$: $\pre{t} \neq \emptyset$ and $\post{t} \neq \emptyset$.
\end{definition}

Well-formed Petri nets are strongly-connected and therefore also proper. Workflow nets are not strongly-connected, but by definition proper. For the main results in this paper, 
we consider proper Petri nets
instead of enforcing stronger structural or behavioral requirements such as strongly-connectedness, liveness, and boundedness.

\section{Lucent Petri Nets}
\label{sec:lucency}

This paper focuses on \emph{lucent} process models whose states are uniquely identified based on the activities they enable.
Lucency is a generic property that can be formulated in the context of Petri nets.
Given a marked Petri net, we would like to know whether each reachable marking has a unique ``footprint'' in terms of the transitions it enables.
If this is the case, then the Petri net is \emph{lucent}.

\begin{definition}[Lucent Petri Nets]\label{def:lucent}
Let $(N,M)$ be a marked Petri net. $(N,M)$ is \emph{lucent} if and only if for any $M_1,M_2 \in R(N,M)$: $\mi{en}(N,M_1)=\mi{en}(N,M_2)$ implies $M_1 = M_2$.
\end{definition}

$(N_1,M_1)$ depicted in Figure~\ref{fig-intro-fc} is lucent.
$(N_2,M_2)$ and $(N_3,M_3)$ in Figures~\ref{fig-intro-nfc} and \ref{fig-locally-safe-not-perpetual} are not lucent.
$(N_4,M_4)$ depicted in Figure~\ref{fig-fc-nonlucid} is also not lucent. Both $[p3,p5,p7]$ and $[p3,p7,p8]$ are reachable from the initial marking and enable the same set of transitions.
\begin{figure}[thb!]
\centering
\includegraphics[width=10.0cm]{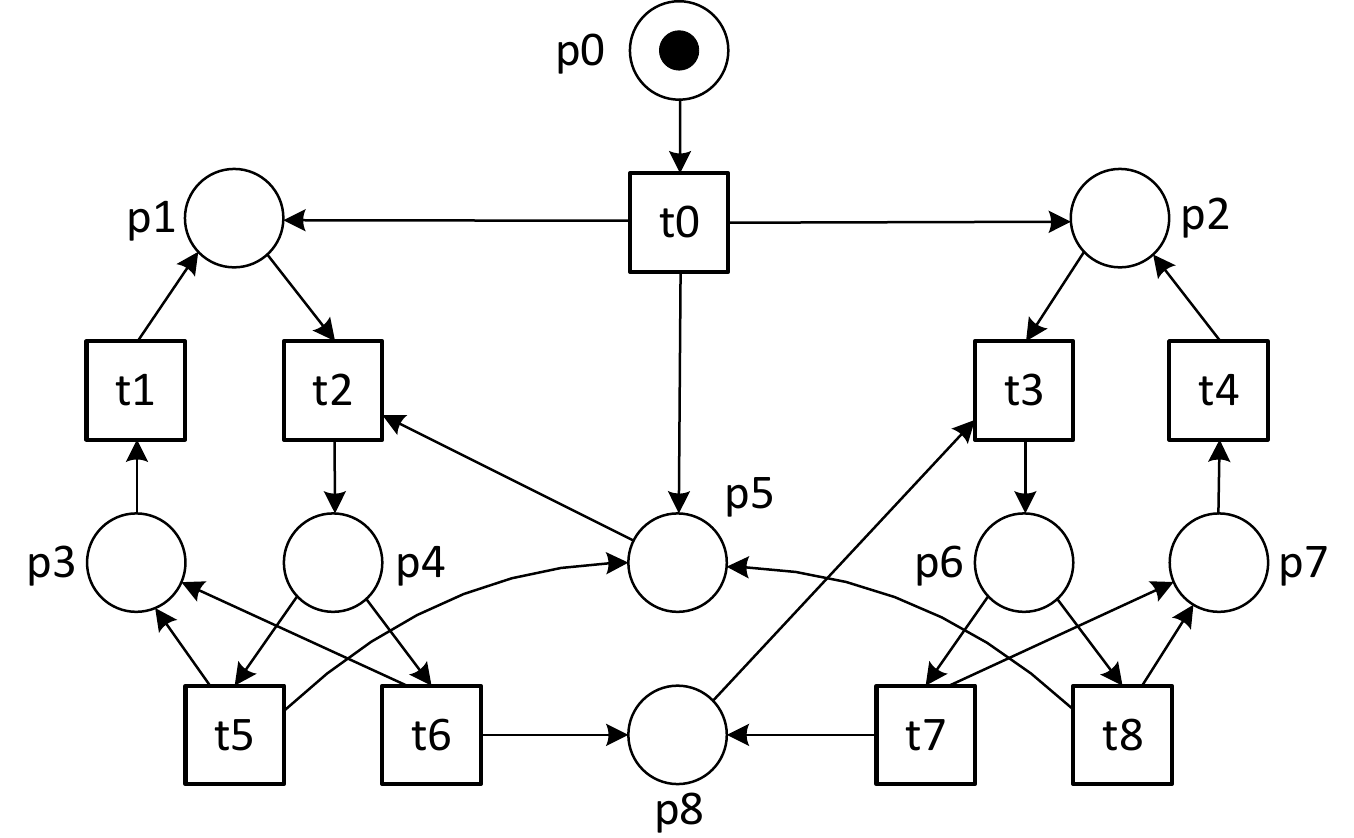}
\caption{$(N_4,M_4)$ is a marked free-choice Petri net that is not lucent: $[p3,p5,p7]$ and $[p3,p7,p8]$ enable $t1$ and $t4$.}
\label{fig-fc-nonlucid}
\end{figure}

Unbounded marked Petri nets are, by definition, not lucent.
However, the examples illustrate that the reverse does not hold.

\begin{proposition}[Boundedness of Lucent Petri Nets]
Any lucent marked Petri net is bounded.
\end{proposition}
\begin{proof}
A marked net with $n= \card{T}$ transitions cannot have more than $2^n$ possible sets of enabled transitions. 
Lucency implies that each set of enabled transitions corresponds to a unique marking. 
Hence, there cannot be more than $2^n$ reachable markings (implying boundedness).
\end{proof}

We would like to find subclasses of nets that are guaranteed to be lucent based on their structure.
At first, one is tempted to think that bounded free-choice nets are lucent.
However, as Figure~\ref{fig-locally-safe-not-perpetual} and Figure~\ref{fig-fc-nonlucid} show, this is not sufficient.

Lucency is related to the notion of transparency, i.e., 
all tokens are in the input places of enabled transitions and therefore not ``hidden''.

\begin{definition}[Transparent Marking]\label{def:trabsp}
Let $(N,M)$ be a marked Petri net. Marking $M$ is a \emph{transparent} marking of $N$ if and only if $M = 
[ p \in P \mid \exists_{t \in \mi{en}(N,M)} \ p \in \pre{t}]$.
$(N,M)$ is fully transparent if and only if each reachable marking is transparent. 
\end{definition}

Full transparency implies lucency, but the reverse does not hold.
Actually, full transparency does not allow for synchronization and concurrency and is therefore very limiting.

\begin{proposition}
Let $(N,M)$ be a marked Petri net.
If $(N,M)$ is fully transparent, then $(N,M)$ is lucent.
The reverse does not hold.
\end{proposition}

Figure~\ref{fig-home-lucent}
shows a marked free-choice Petri net that is lucent but not fully transparent.
Consider, for example, the reachable marking $[p4,p7]$ enabling $t5$. There is only one reachable marking which enables only $t5$.
However, marking $[p4,p7]$ is not transparent since the token in $p7$ is ``hidden''.
\begin{figure}[thb!]
\centering
\includegraphics[width=9.0cm]{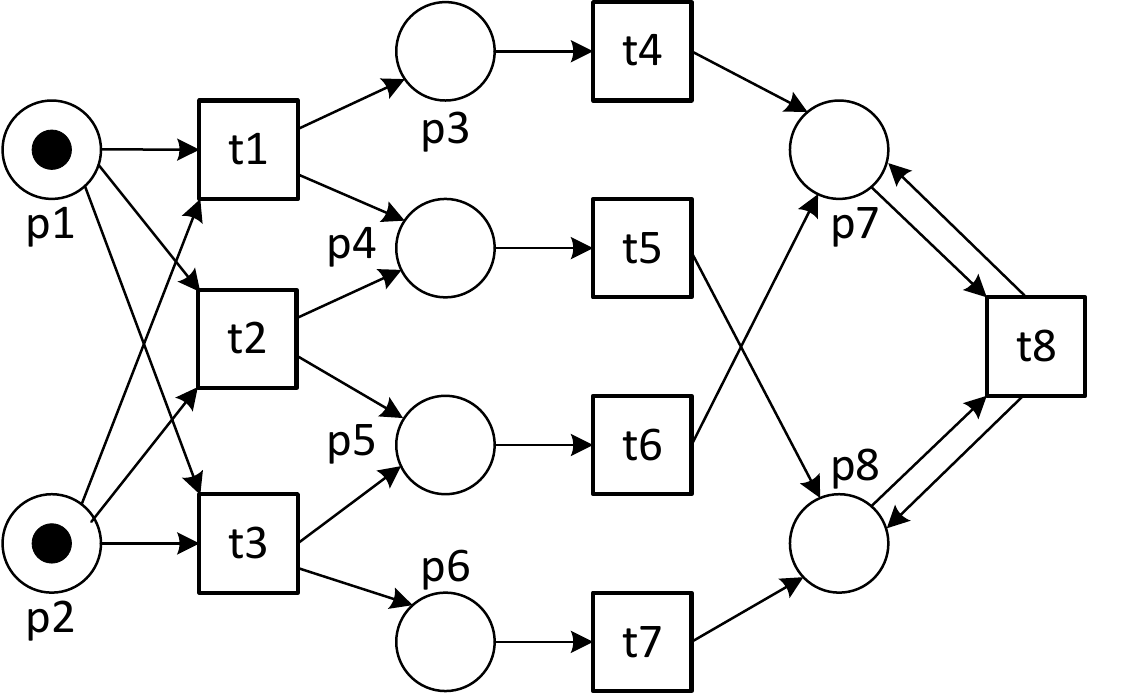}
\caption{$(N_5,M_5)$ is a marked free-choice Petri net that is lucent but not fully transparent.}
\label{fig-home-lucent}
\end{figure}

\section{Free-Choice Nets With Home Clusters Are Lucent}
\label{sec:arelucent}

In \cite{lucent-PN2018}, we defined the class of perpetual nets in an attempt to identify a large class of lucent Petri nets. 
Here, we aim to substantially extend the class of Petri nets 
that is guaranteed to be lucent.
Like in \cite{lucent-PN2018} we use the notion of \emph{home clusters},
but drop the liveness and boundedness requirements.
Actually, none of the Petri nets shown in this paper is perpetual, including the two lucent nets $(N_1,M_1)$ and $(N_5,M_5)$.

\begin{definition}[Home Clusters]\label{def:homeclust}
Let $(N,M)$ be marked Petri net. $C$ is a \emph{home cluster} of $(N,M)$  if and only if $C \in \cluster{N}$ (i.e., $C$ is a cluster) and 
$\mi{Mrk}(C)$ is a home marking of $(N,M)$. If such a $C$ exists, we say that $(N,M)$ has a home cluster.
\end{definition}

Note that a home marking may be dead, but then it should be unique, i.e., a clear \emph{termination point}.
If the initial marking is a home marking, it can be seen as a \emph{regeneration point}.

A mentioned before, 
the key results in this paper apply only to proper Petri nets where all
transitions have input and output places. 
It is always possible to add a self-loop place to ensure this (without changing the behavior).
Moreover, a Petri net having a transition without any input places 
and at least one output place, is unbounded for any initial marking 
and therefore non-lucent.
Transitions without output places make most sense in unbounded nets (which are non-lucent).
Adding a self-loop place to make the Petri net proper, typically results in a model that has no home cluster. However, such models tend to be unbounded and therefore non-lucent anyway.

Note that in literature most authors consider well-formed Petri nets. These are strongly-connected and therefore also proper. Here, we consider a substantially larger class of models.
For example, 
the marked nets $(N_1,M_1)$, $(N_2,M_3)$, 
$(N_4,M_4)$, and $(N_5,M_5)$ are not well-formed, but proper. 
Actually, $(N_3,M_3)$ in Figure~\ref{fig-locally-safe-not-perpetual} is the only well-formed net in this paper (and therefore also proper).
This paper shows that we can drop the well-formedness requirement
and still ensure lucency.

\subsection{Properties of Home Clusters}
\label{sec:prophc}

We first explore some of the essential properties of home clusters
in marked proper free-choice nets.
First, we show that there are two types of clusters: (1) just an isolated end place or (2) a set of places sharing one or more output transitions.

\begin{proposition}[Two Types Of Clusters]\label{prop:cluschar}
Let $(N,M)$ be a marked proper Petri net having a home cluster $C$.
If there is a reachable marking $M' \in R(N,M)$ that is dead, 
then $M' = \mi{Mrk}(C)$, 
$\card{\mi{Pl}(C)} = 1$, and $\mi{Tr}(C) = \emptyset$.
If $(N,M)$ is deadlock-free, then $\mi{Tr}(C) \neq \emptyset$.
\end{proposition}
\begin{proof}
From any reachable marking, one can reach $\mi{Mrk}(C)$.
Hence, if there is a dead marking, then $\mi{Mrk}(C)$ can be the only reachable dead marking.
If not, $\mi{Mrk}(C)$ would not be reachable from this alternative dead marking.
If all places in $\mi{Pl}(C)$ are marked, all transitions $\mi{Tr}(C)$ must be enabled. Hence, $\mi{Tr}(C) = \emptyset$ (otherwise $\mi{Pl}(C)$ cannot be dead).
If $\mi{Tr}(C) = \emptyset$, then the cluster must be a singleton, i.e., $C = \{p_C\}$ (transitions are needed to enlarge the cluster, see Definition~\ref{def:clust}).
If $(N,M)$ is deadlock-free, $\mi{Mrk}(C)$ can be reached and should not be dead. Hence, $\mi{Tr}(C) \neq \emptyset$.
\end{proof}

Most of the results for home markings only apply to \emph{well-formed} free-choice nets,
e.g., 
S-Coverability Theorem,
T-Coverability Theorem,
Rank Theorem,
Duality Theorem, 
Completeness of Reduction Rules Theorem,
Existence of Home Markings Theorem,
Blocking Marking Theorem,
and Home Marking Theorem
\cite{bestfcn,structure-theory-ToPNoC-advanced-course2010,deselesparza}.
We focus on proper free-choice nets
and do \emph{not} require liveness to ensure boundedness, as is shown next.

\begin{definition}[Expedite a Transition in a Transition Sequence]\label{def:expedite}
Let $N=(P,T,F)$ be a free-choice net, $M \in \bag(P)$, 
$\sigma = \langle t_1,t_2, \ldots ,t_i, \ldots ,t_j,\allowbreak \ldots ,t_n \rangle \in T^*$, 
$(N,M)[\sigma\rangle$ (i.e., the sequence $\sigma$ is enabled),
and $1 \leq i < j \leq n$. 
$\mi{exp}_{(N,M)}(\sigma,i,j) = \mi{true}$ if and only if 
\begin{itemize}[noitemsep,topsep=2pt]
\item $(N,M)\allowbreak[\langle t_1,t_2, \ldots ,t_{i-1},t_j\rangle \rangle$ (i.e., it is possible to execute the prefix involving the first $i-1$ transitions followed by $t_j$),  and
\item  $\cluster{t_k} \neq \cluster{t_j}$ for all $k \in \{i, \ldots, j-1\}$ (i.e.,  $t_j$ is the first transition of the respective cluster after $t_{i-1}$).
\end{itemize}
$\mi{exp}_{(N,M)}(\sigma,i,j)$ denotes that the $j$-th transition can be \emph{expedited} by moving $t_j$ to position $i$.
$\sigma_{i \leftarrow j} = \langle t_1,t_2, \ldots ,t_{i-1},t_j,t_i, \ldots ,t_{j-1},t_{j+1} \ldots ,t_n \rangle$ is the corresponding transition sequence where the 
$j$-th transition is moved to the $i$-th position. 

$\mi{Exp}_{(N,M)}(\sigma) \subseteq T^*$ is the subset of all transition sequences that can be obtained by repeatedly expediting transitions, i.e.,
$\mi{Exp}_{(N,M)}(\sigma)$ is the smallest set such that:
\begin{itemize}[noitemsep,topsep=2pt]
\item $\sigma \in \mi{Exp}_{(N,M)}(\sigma)$ and
\item $\sigma'_{i \leftarrow j} \in \mi{Exp}_{(N,M)}(\sigma)$ if $\sigma' \in \mi{Exp}_{(N,M)}(\sigma)$, $1 \leq i < j \leq \card{\sigma'}$, and $\mi{exp}_{(N,M)}(\sigma',i,j)$.
\end{itemize}
\end{definition}

Any $\sigma' \in \mi{Exp}_{(N,M)}(\sigma)$ is a permutation of $\sigma$ and, as we will show next, is enabled if $\sigma$ is enabled.
$\sigma_{i \leftarrow j}$ moves the $j$-th transition to the $i$-th position and is enabled at that position.
Consider $(N_5,M_5)$ in Figure~\ref{fig-home-lucent} and $\sigma = 
\langle t_2,t_5,t_6,t_8,t_8 \rangle$, $\sigma_{2 \leftarrow 3} = \langle t_2,t_6,t_5,t_8,t_8 \rangle$ and $\sigma_{2 \leftarrow 4} = \sigma_{2 \leftarrow 5} = \langle t_2,t_8,t_5,t_6,t_8 \rangle$.
$\sigma_{2 \leftarrow 3} \in \mi{Exp}_{(N_5,M_5)}(\sigma)$, because $\langle t_2,t_6\rangle$ is possible and $t_6$ is the first transition of the respective cluster. 
$\sigma_{2 \leftarrow 4} \not\in \mi{Exp}_{(N_5,M_5)}(\sigma)$, because $\langle t_2,t_8\rangle$ is not possible ($t_8$ is not enabled yet).
Next, we show that expediting transitions is possible and leads to the same marking.

\begin{lemma}[Expediting Transitions Is Safe]\label{lemma:reord}
Let $N=(P,T,F)$ be a free-choice net, 
$M,M' \in \bag(P)$, and
$\sigma \in T^*$ such that
$(N,M)[\sigma \rangle (N,M')$.
For any $\sigma' \in \mi{Exp}_{(N,M)}(\sigma)$:
$(N,M)[\sigma' \rangle (N,M')$.
\end{lemma}
\begin{proof}
Assume $N=(P,T,F)$ is a free-choice net and $M$, $M'$, and $\sigma$ are such that $(N,M)[\sigma \rangle (N,M')$.
$\mi{Exp}_{(N,M)}(\sigma)$ is defined as the smallest set such that (1) $\sigma \in \mi{Exp}_{(N,M)}(\sigma)$ and (2)
$\sigma'_{i \leftarrow j} \in \mi{Exp}_{(N,M)}(\sigma)$ if $\sigma' \in \mi{Exp}_{(N,M)}(\sigma)$, $1 \leq i < j \leq \card{\sigma}$, and $\mi{exp}_{(N,M)}(\sigma',i,j)$.
We provide a proof using induction based on the iterative construction of $\mi{Exp}_{(N,M)}(\sigma)$.

(1) The base step  $\sigma'= \sigma$  obviously holds, because $\sigma \in \mi{Exp}_{(N,M)}(\sigma)$ and $(N,M)[\sigma \rangle (N,M')$.

(2) For the inductive step, it suffices to prove that
$(N,M)[\sigma'_{i \leftarrow j} \rangle (N,M')$ assuming that $\sigma' = \langle t_1, \ldots ,t_n \rangle \in \mi{Exp}_{(N,M)}(\sigma)$, $1 \leq i < j \leq n$, $\mi{exp}_{(N,M)}(\sigma',i,j)$, and $(N,M)[\sigma' \rangle (N,M')$. 
We need to prove that $\sigma'_{i \leftarrow j} = \langle t_1,t_2, \ldots ,t_{i-1},\allowbreak t_j,\allowbreak t_i, \ldots ,t_{j-1},t_{j+1} \ldots ,t_n \rangle$ is indeed enabled
and leads to the same final marking, i.e.,  $(N,M)[\sigma'_{i \leftarrow j} \rangle (N,M')$. 
Let $M''$ be the marking after firing the first $i-1$ transitions, i.e., $(N,M)\allowbreak[\langle t_1,t_2, \ldots ,t_{i-1}\rangle \rangle (N,M'')$.
$t_j \in \mi{en}(N,M'')$ because $\mi{exp}_{(N,M)}(\sigma',i,j)$ (see first condition).
The transitions $t_i, \ldots ,t_{j-1}$ do not consume any tokens from $\cluster{t_j}$ (use the second condition in $\mi{exp}_{(N,M)}(\sigma',i,j)$ stating that $\cluster{t_k} \neq \cluster{t_j}$ for all $k \in \{i, \ldots, j-1\}$) 
and therefore can still be executed ($t_j$ only consumed tokens from places in $\cluster{t_j}$).
The marking reached after $\langle t_1,t_2, \ldots ,t_{i-1},\allowbreak t_j,\allowbreak t_i, \ldots ,t_{j-1} \rangle$ is the same as
reached after prefix $\langle t_1,t_2, \ldots ,t_{i-1},\allowbreak t_i, \ldots ,t_{j-1},t_j \rangle$. 
Moreover, the same subsequence of transitions $\langle t_{j+1} \ldots ,t_n \rangle$ remains.
Hence, $(N,M)[\sigma'_{i \leftarrow j} \rangle (N,M')$ thus completing the proof.
\end{proof}

Note that for any $\sigma' \in \mi{Exp}_{(N,M)}(\sigma)$: $[t \in \sigma'] = [t \in \sigma]$ (i.e., $\sigma'$ and $\sigma$ are permutations of the same multiset) and the order per cluster does not change, i.e., transitions can only ``overtake'' transitions of other clusters.
Lemma~\ref{lemma:reord} shows that expediting transitions does not jeopardize the ability to execute the remainder of an enabled firing sequence. 
This can be used to show that it is impossible to have a marking dominating the home marking (i.e., one cannot reach a strictly larger marking). 

\begin{theorem}[No Dominating Markings in Free-Choice Nets With a Home Cluster]\label{theo:dommark}
Let $(N,M)$ be a marked proper free-choice net having a home cluster $C$.
For all $M' \in R(N,M)$: if $M' \geq \mi{Mrk}(C)$, 
then $M' = \mi{Mrk}(C)$.
\end{theorem}
\begin{proof}
Consider a marked proper free-choice net $(N,M)$ having a home cluster $C$.
Assume there exists a reachable marking $M' \in R(N,M)$ such that $M' > \mi{Mrk}(C)$.
We show that this is \emph{impossible}, thereby proving the theorem.

First, we assume that $(N,M)$ has a deadlock and show that this leads to a contradiction.
Using Proposition~\ref{prop:cluschar}, 
we know that $\mi{Mrk}(C)$ is the only reachable dead marking 
and $\mi{Tr}(C) = \emptyset$. 
However, $M' > \mi{Mrk}(C)$ is reachable and the token in $C$ cannot be removed anymore if $\mi{Tr}(C) = \emptyset$.
Since the net is proper, any marking reachable from $M'$ will have at least one extra token next to the token in $\mi{Mrk}(C)$.
Therefore, $\mi{Mrk}(C)$ cannot be reached, contradicting that $C$ is a home cluster. Hence, $(N,M)$ must be deadlock-free. 

Since $(N,M)$ is deadlock-free, $\mi{Tr}(C) \neq \emptyset$ (use Proposition~\ref{prop:cluschar}), i.e., the home cluster has at least one transition.
All transitions in $\mi{Tr}(C)$ live, because we can always reach the home marking  $\mi{Mrk}(C)$ again and again.

Without loss of generality, we can assume that $M' \in R(N,M)$ is such that 
the distance to the home marking $\mi{Mrk}(C)$ is \emph{minimal}.
Let $\sigma_s$ be a \emph{shortest} path from $M'$ to $\mi{Mrk}(C)$ having length $ \card{\sigma_s}$.
In other words, $(N,M')[\sigma_s\rangle (N,\mi{Mrk}(C))$, and for all $M_{\mi{alt}} \in R(N,M)$ and $\sigma_{\mi{alt}} \in T^*$ such that $M_{\mi{alt}} > \mi{Mrk}(C)$ and 
$(N,M_{\mi{alt}})[\sigma_{\mi{alt}}\rangle (N,\allowbreak \mi{Mrk}(C))$: $\card{\sigma_{\mi{alt}}} \geq  \card{\sigma_s}$. Obviously, $\card{\sigma_s} \geq 1$ (otherwise $M' = \mi{Mrk}(C)$ contradicting our initial assumption).

$\mi{Exp}_{(N,M')}(\sigma_s)$ contains all permutations 
of the shortest sequence $\sigma_s$ that are obtained by expediting transitions.
Let $\sigma_1$ and $\sigma_2$ be such that $\sigma_1 \cdot \sigma_2 \in \mi{Exp}_{(N,M')}(\sigma_s)$, $(N,\mi{Mrk}(C))[\sigma_1\rangle (N,M'')$, and
$\mi{en}(N,M'') \cap \{t \in \sigma_2\} = \emptyset$.
$\sigma_1$ contains the transitions in $\sigma_s$ that can also be executed
starting from the home marking. This leads to marking $M''$. In this marking, none of the remaining transitions in $\sigma_s$ (i.e., the transitions in $\sigma_2$) can be executed.
In other words, starting from we $\mi{Mrk}(C)$, we try to execute as much of $\sigma_s$ as possible by expediting transitions (as described in Definition~\ref{def:expedite}). $\sigma_1$ is the part that can be executed
(leading to $M''$) and $\sigma_2$ is the remaining part of $\sigma_s$. $\sigma_2$ only contains transitions that are not enabled in $M''$. Obviously, there always exist
$\sigma_1$ and $\sigma_2$ such that these requirements are met ($\mi{Exp}_{(N,M')}(\sigma_s) \neq \emptyset$ and we can add transitions to $\sigma_1$ until this is no longer possible). 
Moreover, $\sigma_1$ can also be executed starting in $M'$ because it is the prefix of an expedited sequence. Let $M'''$ be the corresponding marking, i.e., 
$(N,M')[\sigma_1\rangle (N,M''')$. From this marking, we can reach the home marking
by executing $\sigma_2$ (because $\sigma_1 \cdot \sigma_2 \in \mi{Exp}_{(N,M')}(\sigma_s)$ and Lemma~\ref{lemma:reord}).
\begin{figure}[thb!]
\centering
\includegraphics[width=5.5cm]{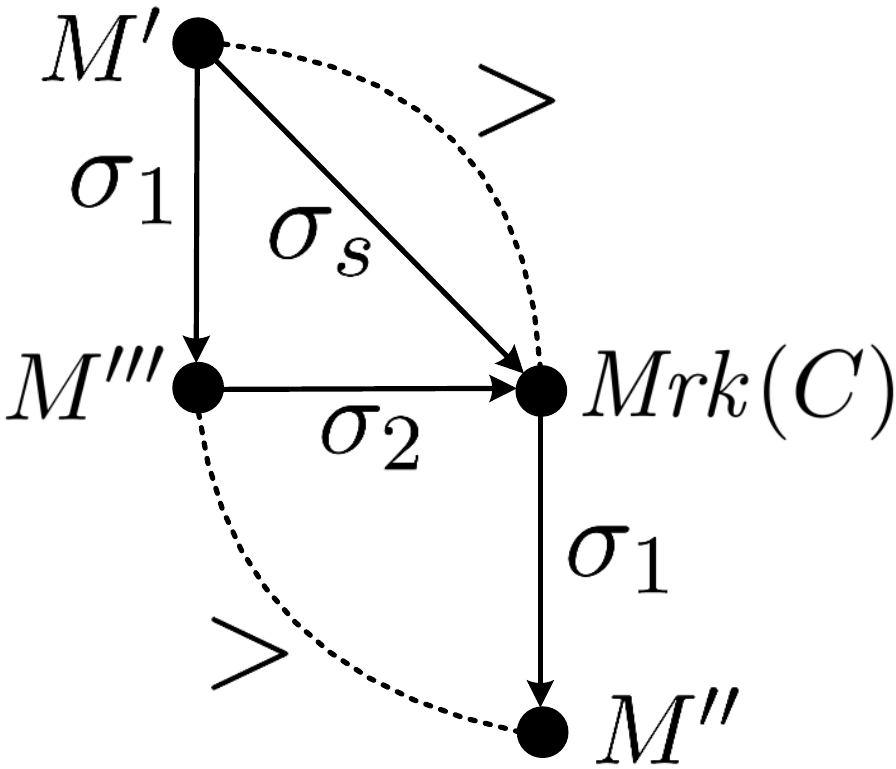}
\caption{Sketch of the construction used in Theorem~\ref{theo:dommark}.
The solid arrows denote firing sequences and the dashed lines indicate multiset domination.
 $\sigma_s$ is a firing sequence of minimal length leading from a marking $M'$ (which is larger than $\mi{Mrk}(C)$) to $\mi{Mrk}(C)$. Firing sequence $\sigma_1 \cdot \sigma_2$ is a permutation of $\sigma_s$ such that the transitions also enabled when starting from $\mi{Mrk}(C)$ are expedited leading to firing sequence $\sigma_1$. The remaining transitions in $\sigma_2$ are not enabled when starting from $\mi{Mrk}(C)$.}
\label{fig-proof-dom}
\end{figure}

To summarize, 
$\sigma_1 \cdot \sigma_2 \in \mi{Exp}_{(N,M')}(\sigma_s)$,
$(N,\mi{Mrk}(C))[\sigma_1\rangle (N,M'')$,
$(N,M')[\sigma_1\rangle (N,M''')$,
$(N,M''')[\sigma_2\rangle (N,\mi{Mrk}(C))$, and
$\mi{en}(N,M'') \cap \{t \in \sigma_2\} = \emptyset$.
Moreover, because $M' > \mi{Mrk}(C)$ also $M''' > M''$ and $\card{\sigma_s} \geq 1$.
Figure~\ref{fig-proof-dom} shows the relations between the different markings. 
To complete the proof we consider two cases ($\sigma_1 = \langle ~ \rangle$ and $\sigma_1 \neq \langle ~ \rangle$):
\begin{itemize}[noitemsep,topsep=2pt]
\item Assume $\sigma_1 = \langle ~ \rangle$. 
This implies that $M''' = M'$, $M'' = \mi{Mrk}(C)$, $\sigma_2 = \sigma_s$, $\mi{en}(N,M'')= \mi{Tr}(C)$,
and $\mi{Tr}(C) \cap \{t \in \sigma_s\} =\emptyset$.
Hence, when executing $\sigma_s$ starting from $M'$ the home cluster remains fully marked.  
However, there is at least one additional token in $M'$ that cannot ``disappear''
when executing $\sigma_s$ (the net is proper) leading to a contradiction. 
\item Assume $\sigma_1 \neq \langle ~ \rangle$. 
This implies that $M''' \not> \mi{Mrk}(C)$, otherwise there would be a shorter
sequence than $\sigma_s$, namely $\sigma_2$. (Recall that we selected $M'$ and $\sigma_s$ such that there is no shorter sequence leading to the home marking.)
There must exist an enabled cluster $C_e$ in $M''$ (the net cannot be dead), i.e., 
$\mi{Tr}(C_e) \subseteq \mi{en}(N,M'')$.
$\mi{Tr}(C_e) \cap \{t \in \sigma_2\} = \emptyset$.
If $C_e = C$, then we find a contradiction 
because this implies $M''' > \mi{Mrk}(C)$.
If $C_e \neq C$, then there is a place 
$p_e \in \mi{Pl}(C_e)$ outside of the home cluster $C$ 
that is marked in $M''$ and also $M'''$,
but $p_e \not\in \mi{Mrk}(C)$.
The token in $p_e$ is never removed by the transitions in $\sigma_2$.
However, after executing $\sigma_2$, place $p_e$ should be empty because only places in $C$ are marked, 
thus leading to a contradiction.
\end{itemize}
Hence, in all cases we find a contradiction, proving that $M' \leq \mi{Mrk}(C)$.
\end{proof}

Theorem~\ref{theo:dommark} implies boundedness.
Later, we show that marked proper free-choice nets having a home cluster are also safe.

\begin{corollary}[Boundedness]\label{corr:dommark}
Let $(N,M)$ be a marked proper free-choice net having a home cluster $C$.
For all $M_1,M_2 \in R(N,M)$: $M_1 \not> M_2$.
Hence, $(N,M)$ is also bounded.
\end{corollary}
\begin{proof}
Assume $M_1,M_2 \in R(N,M)$ such that $M_1 > M_2$ (first marking is strictly larger).
There exists a $\sigma$ such that
$(N,M_2)[\sigma\rangle (N,\mi{Mrk}(C))$.
Since $M_1 > M_2$ there must be another reachable marking $M_3$ such that
$(N,M_1)[\sigma\rangle (N,M_3)$ and $M_3 > \mi{Mrk}(C)$.
However, Theorem~\ref{theo:dommark} says this is impossible, leading to a contradiction.
\end{proof}

\subsection{Rooted Disentangled Paths}
\label{sec:rdespath}

We will now reason about the numbers of tokens on specific paths in the Petri net. Therefore, we first provide some standard definitions and then introduce the new notion of \emph{rooted disentangled paths}.

\begin{definition}[Elementary Paths and Circuits]\label{def:elempath}
A \emph{path} in a Petri net $N=(P,T,F)$ is a non-empty sequence of nodes $\rho = \langle x_1,x_2, \ldots ,x_n \rangle$ such that $(x_i,x_{i+1}) \in F$ for $1 \leq i < n$.
Hence, $x_{i-1} \in {\pre{x_i}}$ for $1< i \leq n$ and $x_{i+1} \in {\post{x_i}}$ for $1 \leq i < n$.
$\mi{paths}(N)$ is the set of all paths in $N$.
$\rho$ is an \emph{elementary path} if $x_i \neq x_j$ for $1 \leq i < j \leq n$ (i.e., no element occurs more than once). An elementary path is called is a \emph{circuit} if $x_1 \in {\post{x_n}}$.
\end{definition}

Next, we focus on paths that start and end with a place and that visit a cluster at most once.
Consider $(N_5,M_5)$ in Figure~\ref{fig-home-lucent}.
$\langle t8,p7,t8,p8\rangle$ is a path that is not elementary. This implies that also a cluster is visited multiple times.
$\langle p1,t1,p3,t4,p7 \rangle$ is a so-called disentangled path since each place on the path belongs to a different cluster.

\begin{definition}[(Rooted) Disentangled Paths]\label{def:rdespath}
Let $N=(P,T,F)$ be a Petri net.
$\rho = \langle p_1,t_1,p_2, \ldots , t_{n-1},p_n \rangle$ is a \emph{disentangled path} of $N$ if and only if 
$\rho$ is a path of $N$ ($\rho \in \mi{paths}(N)$),
$p_1 \in P$, 
$p_n \in P$, and
for all $1 \leq i < j \leq n$: $\cluster{p_i} \neq \cluster{p_j}$ (i.e., $\rho$ starts and ends with a place and does not contain elements that belong to the same cluster).
A disentangled path is $Q$-\emph{rooted} if $p_n \in Q$.
\end{definition}

Disentangled paths are elementary, but not all elementary paths are disentangled.
Consider $N_3$ in Figure~\ref{fig-locally-safe-not-perpetual}.
$\rho_1 = \langle p5,t3,p3,t2,p4 \rangle$ is elementary, but not disentangled because $\cluster{p5} = \cluster{p4}$.
$\rho_2 = \langle p5,t3,p3,t2,p1 \rangle$ is elementary and disentangled.
$\rho_2$ is $Q$-rooted where $Q$ can be any subset of places that includes $p1$.

In the remainder of this subsection, we reason about the existence of disentangled paths and the number of tokens on them.

\begin{lemma}[Existence of Rooted Disentangled Paths]\label{lem:exstrdespath}
Let $N=(P,T,F)$ be a free-choice net, $C$ a cluster of $N$, $p \in P$, and $q \in C \cap P$. If $N$ has a path $\rho \in \mi{paths}(N)$ starting in $p$ and ending in $q$, then there also exists a $C$-rooted disentangled path starting in $p$.
\end{lemma}
\begin{proof}
Let $\rho = \langle p_1,t_1,p_2, \ldots , t_{n-1},p_n \rangle \in \mi{paths}(N)$ be the path connecting $p = p_1$ and $q = p_n$.
$\rho$ can be converted into a $Q$-rooted disentangled path starting in $p$. 
This is done by removing parts of the path through shortcuts that can be taken when the same cluster is visited multiple times.
Let $i \in \{1, \ldots, n\}$ be a pointer pointing to place $p_i$ in $\rho$.
We start with $i = 1$ (i.e., $i$ points to the first place $p_1$) and move towards the end of the path $i=n$. 
\begin{itemize}[noitemsep,topsep=2pt]
\item If pointer $i$ points to place $p_i$ and $p_i \in C$, we can ignore the rest of the sequence because we already reached $C$ via a unique sequence of clusters. 
(Note that if $p=p_1$ is already in $C$, we have a sequence of length 1.)
\item If $p_i \not\in C$, then $i < n$ because $q = p_n \in C$. 
Hence, there still exists an output transition $t_i$ with output place $p_{i+1}$ in $\rho$.
\begin{itemize}[noitemsep,topsep=2pt]
\item If none of the $p_j$ with $j>i$ is in the same cluster as $p_i$, then we keep $p_i$ and $t_i$, and continue with $p_{i+1}$ (i.e, increment $i$). 
\item If there is a $p_j$ with $j>i$ that is in the same cluster as $p_i$, then we take the largest $j$ for which this holds. Also $j < n$, because $p_j \not\in C$ (here $p_i$ and $p_j$ are in the same cluster). Hence, there exists a $t_j$ and $p_{j+1}$.
$p_i$ and $p_j$ may refer to the same place or not.
However, $\{p_i,p_j\} \subseteq \pre{t_j}$ (both are in the same cluster and all transitions in the cluster consume from all places in the cluster).
Since $p_i \in \pre{t_j}$, we can remove the subsequence $\langle t_i,p_{i+1}, \ldots , t_{j-1},p_j \rangle$ and directly connect $p_i$ to $t_j$. 
Sequence 
$\langle \ldots , p_i, t_j, p_{j+1}, \ldots , t_{n-1},\allowbreak p_n \rangle$
constitutes a path in the Petri net and we continue with $p_{j+1}$ (i.e, set $i = j+1$).
In summary, $\langle p_1, \ldots , p_i, t_i, t_{i+1}, \ldots, p_j, t_j, p_{j+1}, \ldots , t_{n-1},p_n \rangle$ is transformed into
$\langle p_1, \ldots , p_i, \allowbreak \stkout{t_i, t_{i+1}, \ldots, p_j}, t_j, p_{j+1}, \ldots , t_{n-1},p_n \rangle$ and continues with $i=j+1$.
\end{itemize}
\end{itemize}
We repeat this process until we reach $C$.
Each of the elements in the resulting path is connected to the previous one and we never visit the same cluster twice. We also keep the initial place $p$. Therefore, the resulting path is a $C$-rooted disentangled path starting in $p$.
\end{proof}

Consider the path $\rho = \langle p6,t4,p5,t3,p3,t2,p4,t3,p3,t2,p1 \rangle$ in Figure~\ref{fig-locally-safe-not-perpetual}.
the path ends in the cluster $C=\{p1,t1\}$. Using the approach used in the proof of Lemma~\ref{lem:exstrdespath}, 
this path is converted into the $C$-rooted disentangled path $\langle p6,t4,p5,\stkout{t3,p3,t2,p4},t3,p3,t2,p1 \rangle = \langle p6,t4,p5,t3,p3,t2,p1 \rangle$.
We can construct a $C$-rooted disentangled path starting in any place $p$ that is not dead,
i.e., a place marked in at least one of reachable markings.

\begin{corollary}[Existence of Rooted Disentangled Paths from Marked Places]\label{corr:pathexists}
Let $(N,M)$ be a marked proper free-choice net having a home cluster $C$.
For any non-dead place $p \in P$, there exists a $C$-rooted disentangled path starting in $p$.
\end{corollary}
\begin{proof}
Take an arbitrary place $p$ that can be marked in some reachable marking $M'$. If $p \in C$, then $\rho = \langle p \rangle$ is a
$C$-rooted disentangled path.
If $p \not\in C$, then there must be a path from $p$ to one of the places in $C$ (say $q$). 
This follows from the fact that the net is proper, i.e., for all $t \in T$: $\pre{t} \neq \emptyset$ and $\post{t} \neq \emptyset$.
Therefore, a token can not simply disappear and must end up in $C$.
To see this, color the token in $p$ red and then execute a firing sequence ending in $\mi{Mrk}(C)$. 
When executing a transition with at least one red token, make all produced tokens also red.
Because we cannot consume a red token without producing at least one new red token, 
we know that at least one red token will end up in $\mi{Mrk}(C)$.
This proves that there is a path $\rho$ starting in $p$ and ending in some $q \in C \cap P$ (follow back one red token in $\mi{Mrk}(C)$).
Since there is such a path $\rho$, there also exists a $C$-rooted disentangled path starting in $p$ 
(apply Lemma~\ref{lem:exstrdespath}).
\end{proof}

Dead places do not change the behavior and can be removed together with the output transitions if desired 
(but do not have to be removed, since they remain empty anyway).
The next lemma plays a key role in our analysis of nets having a home cluster $C$: 
$C$-rooted disentangled paths are \emph{safe}, i.e., 
at any time \emph{all} places on a $C$-rooted disentangled path \emph{together} contain at most one token.

\begin{lemma}[Rooted Disentangled Paths Are Safe]\label{lem:pathsafety}
Let $(N,M)$ be a marked proper free-choice net having a home cluster $C$.
For any reachable marking, $M' \in R(N,M)$ and $C$-rooted disentangled path $\rho = \langle p_1,t_1,p_2, \ldots , t_{n-1},\allowbreak p_n \rangle$: $M'(\{p_1,p_2, \ldots ,p_n\}) \allowbreak \leq 1$.
\end{lemma}
\begin{proof}
Assume $(N,M)$ is a marked proper free-choice net, $C$ is a home cluster,  and $\rho = \langle p_1,t_1,p_2, \ldots , t_{n-1},\allowbreak p_n \rangle$ is a $C$-rooted disentangled path.
Let $P_{\rho} =  \{p_1,p_2, \ldots ,p_n\}$ and $T_{\rho} =  \{t_1,t_2, \ldots ,t_{n-1}\}$.

Assume that the lemma does not hold, i.e., $\rho$ is not safe and
$M'(P_{\rho}) > 1$ for some $M' \in R(N,M)$. 
We show that this leads to a contradiction.

Consider the tokens (at least two) in the places $P_{\rho}$. We try to move these tokens towards $p_n \in C$.
Each place $p_i \in P_{\rho}$ corresponds to a unique cluster $C_i$ because the path is disentangled. This combined with the free-choice property, allows us to fully control the trajectories of the tokens in $P_{\rho}$.

First, we look a the case where $C$ has a transition, say $t_C$ (i.e., there are no dead markings, see Proposition~\ref{prop:cluschar}).
We start in marking $M_c = M'$. If one of the transitions in $T_{\rho}$ is enabled, then we fire this transition (in any order and perhaps also multiple times) and update the current marking $M_c$. This cannot decrease the number of tokens, i.e., we still have $M_c(P_{\rho}) > 1$. 
Note that a transition in $T_{\rho}$ consumes precisely one token
``from the path'' and produces at least one token ``on the path'' (disentangled paths are elementary).
If $t_C$ is enabled, then $M_c \geq \mi{Mrk}(C)$. 
However, given the second token in $P_{\rho}$ this implies $M_c > \mi{Mrk}(C)$. 
This leads to a contradiction using Theorem~\ref{theo:dommark}. 
If none of the transitions in $T_{\rho} \cup \{t_C\}$ is enabled in $M_c$, then
we must be able to fire a sequence of other transitions enabling a transition in $T_{\rho} \cup \{t_C\}$. 
$C$ is a home cluster and there are no deadlocks, so we can always walk towards a marking enabling one of the transitions in 
$T_{\rho} \cup \{t_C\}$. The moment one of the transitions
in $T_{\rho}$ is enabled, we can again control the choices involved.
Hence, we can continue to move tokens along the path until we find a contradiction.

Next, we look a the case where $C$ does not have a transition (i.e., the home marking is is a deadlock, see Proposition~\ref{prop:cluschar}).
We can use exactly the same strategy to move the tokens towards $p_n$ (there one case less to consider). 
The moment a token reaches $p_n$ there is at least one additional token in $P_{\rho}$ and this one can also be moved to $p_n$ 
leading to a contradiction (apply Theorem~\ref{theo:dommark} to show that there cannot be two tokens in $p_n$).
\end{proof}

The previous results can be combined to show
that the class of marked Petri nets considered is safe.

\begin{corollary}[Marked Proper Free-Choice Net Having a Home Cluster Are Safe]\label{corr:safe}
Let $(N,M)$ be a marked proper free-choice net having a home cluster $C$. $(N,M)$ is safe.
\end{corollary}
\begin{proof}
Follows directly from Lemma~\ref{lem:pathsafety} and Corollary~\ref{corr:pathexists}.
If a place $p$ is dead, then it will never have a token and thus safe.
If a place $p$ is not dead, then there exists a $C$-rooted disentangled path (Corollary~\ref{corr:pathexists}) starting in $p$, and this path must be safe due to Lemma~\ref{lem:pathsafety}. Hence, all places are safe.
\end{proof}

\subsection{Conflict-Pairs}
\label{sec:cfpair}

If a marked Petri net is not lucent, then there must be two different markings enabling the same set of transitions.
We will convert such a pair of markings 
into a \emph{conflict-pair}. By showing that these do not exist, we can prove lucency. 

\begin{definition}[Conflict-Pair]\label{def:cfpair}
Let $(N,M)$ be a marked Petri net.
$(M_1,M_2)$ is called a \emph{conflict-pair} for $(N,M)$ 
if and only if
\begin{itemize}[noitemsep,topsep=2pt]
\item $M_1$ and $M_2$ are reachable markings of $(N,M)$ (i.e., $M_1,M_2 \in R(N,M)$),
\item $M_1$ and $M_2$ are not dead (i.e., $\mi{en}(N,M_1) \neq \emptyset$ and $\mi{en}(N,M_2) \neq \emptyset$),
\item $\mi{en}(N,M_1) \cap \mi{en}(N,M_2) = \emptyset$ (no transition is enabled in both markings),
\item for all $t \in \mi{en}(N,M_1)$: $M_2(\pre t)\geq 1$, and
\item for all $t \in \mi{en}(N,M_2)$: $M_1(\pre t)\geq 1$.
\end{itemize}
\end{definition}
Consider $(N_3,M_3)$ Figure~\ref{fig-locally-safe-not-perpetual} and markings $M_1 = [p2, p3, p5]$ and $M_2 = [p2, p4, p5]$.
$M_1$ can be reached by firing $t1$ and $t4$.
$M_2$ can be reached by firing $t1$, $t2$, $t1$, and $t4$.
$\mi{en}(N,M_1) = \{t2\}$, $\mi{en}(N,M_2) = \{t3\}$,
$\mi{en}(N,M_1) \cap \mi{en}(N,M_2) = \emptyset$,
$M_2(\pre t2) = 1 \geq 1$, and
$M_1(\pre t3) = 1 \geq 1$.

To better understand the above definition, let us split
each of the two markings in the conflict-pair $(M_1,M_2)$ in an ``agreement'' and a ``disagreement''  part.
$M^{\mi{agree}}$ is the maximal marking such that 
$M^{\mi{agree}} \leq M_1$ and $M^{\mi{agree}} \leq M_2$.
Now we can write 
$M_1 = M^{\mi{agree}} \bplus M^{\mi{disagree}}_1$ and
$M_2 = M^{\mi{agree}} \bplus M^{\mi{disagree}}_2$.
Obviously, all three submarkings $M^{\mi{agree}}$, $M^{\mi{disagree}}_1$, and $M^{\mi{disagree}}_2$ are non-empty.
This allows us to speak about ``agreement tokens'' (tokens in $M^{\mi{agree}}$) and ``disagreement tokens'' (tokens in $M^{\mi{disagree}}_1$ or $M^{\mi{disagree}}_2$).
For $M_1 = [p2, p3, p5]$ and $M_2 = [p2, p4, p5]$, we have $M^{\mi{agree}} = [p2,p5]$, $M^{\mi{disagree}}_1 = [p3]$, 
and $M^{\mi{disagree}}_2 = [p4]$.

Both $M_1$ and $M_2$ should enable at least one transition, but there cannot be a transition enabled by both.
This means that $\mi{en}(N,M^{\mi{agree}}) = \emptyset$.
The last two requirements in Definition~\ref{def:cfpair}
state that transitions enabled in  $M_1$ and $M_2$ should also
consume at least one agreement token. 
Hence, for any $t_1 \in \mi{en}(N,M_1)$:
$t_1 \not\in \mi{en}(N,M_2)$,
$t_1 \not\in \mi{en}(N,M^{\mi{agree}})$,
$M^{\mi{agree}}(\pre{t_1}) \geq 1$, and
$M^{\mi{disagree}}_1(\pre{t_1}) \geq 1$.
Similarly, for any $t_2 \in \mi{en}(N,M_2)$:
$t_2 \not\in \mi{en}(N,M_1)$,
$t_2 \not\in \mi{en}(N,M^{\mi{agree}})$,
$M^{\mi{agree}}(\pre{t_2}) \geq 1$, and
$M^{\mi{disagree}}_2(\pre{t_2}) \geq 1$.

%

Next, we show that the absence of conflict-pairs implies lucency.
Later, we show that a marked proper free-choice net with a home cluster cannot have a conflict-pair.
Hence, such nets are guaranteed to be lucent.
\begin{figure}[thb!]
\centering
\includegraphics[width=15.0cm]{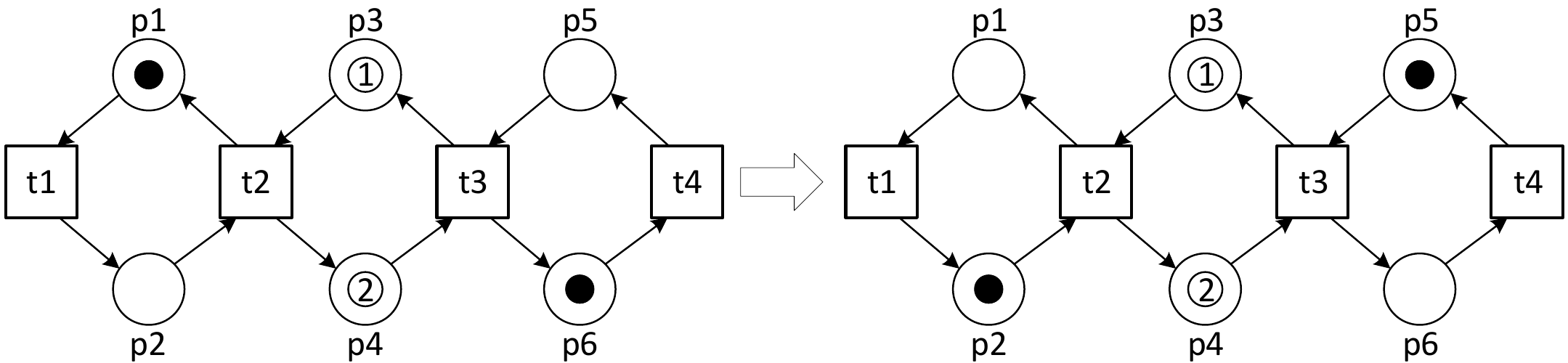}
\caption{Example showing how two markings that have the same ``footprint'' (left) in terms of enabling can be converted into a conflict-pair (right).
The left-hand side shows markings $M_1 = [p1,p3,p6]$ and $M_2 = [p1,p4,p6]$.
The right-hand side shows markings $M_1' = [p2,p3,p5]$ and $M_2' = [p2,p4,p5]$.
The ``agreement tokens'' are depicted as black dots (denoted by $\bullet$), the ``disagreement tokens'' are shown as 
$\raisebox{.5pt}{\textcircled{\raisebox{-.9pt} {1}}}$ (only in $M_1$ and $M_1'$)
or $\raisebox{.5pt}{\textcircled{\raisebox{-.9pt} {2}}}$ (only in $M_2$ and $M_2'$).
}
\label{fig-locally-safe-conflict-pair}
\end{figure}

To show that the absence of conflict-pairs implies lucency for free-choice nets having a home cluster, 
Lemma~\ref{lem:cflucent} shows that it is possible to convert two markings $M_1$ and $M_2$ that have the same ``footprint'' in terms of enabling (i.e., $\mi{en}(N,M_1)=\mi{en}(N,M_2)$) into a conflict-pair $(M_1',M_2')$.
To illustrate the construction, we consider the free-choice net $N_3$ in Figure~\ref{fig-locally-safe-not-perpetual} 
which does \emph{not} have a home cluster 
(we can find two markings having the same ``footprint'' because of this).
The left-hand side of Figure~\ref{fig-locally-safe-conflict-pair} shows 
the markings $M_1 = [p1,p3,p6]$ and $M_2 = [p1,p4,p6]$.
$M_1$ is the initial marking and $M_2$ can be reached by firing $t1$ and $t2$.
Tokens in $M_1$ but not in $M_2$ are represented by $\raisebox{.5pt}{\textcircled{\raisebox{-.9pt} {1}}}$ 
and tokens in $M_2$ but not in $M_1$ are represented by $\raisebox{.5pt}{\textcircled{\raisebox{-.9pt} {2}}}$.
Tokens in both markings are denoted by $\bullet$.
$M_1$ and $M_2$ demonstrate that the net is \emph{not} lucent because $\mi{en}(N_3,M_1)=\mi{en}(N_3,M_2)=\{t1,t4\}$.
To move to the conflict-pair $(M_1',M_2')$ with $M_1' = [p2,p3,p5]$ and $M_2' = [p2,p4,p5]$ on the right-hand side of Figure~\ref{fig-locally-safe-conflict-pair}, we do not ``touch'' the disagreement tokens denoted by 
$\raisebox{.5pt}{\textcircled{\raisebox{-.9pt} {1}}}$ and $\raisebox{.5pt}{\textcircled{\raisebox{-.9pt} {2}}}$.
This implies that no transitions in the corresponding clusters can fire and that these disagreement tokens do not move.
Hence, we can only fire transitions that only consume agreement tokens. These are depicted as normal black dots $\bullet$ in Figure~\ref{fig-locally-safe-conflict-pair}. In the example, we can fire $t1$ and $t4$ involving only agreement tokens. Such transitions consume and produce agreement tokens.
Since the net is guaranteed to be safe, no agreement tokens can be produced for places that have disagreement tokens 
(i.e., $p3$ and $p4$ in Figure~\ref{fig-locally-safe-conflict-pair}).
Hence, the $\raisebox{.5pt}{\textcircled{\raisebox{-.9pt} {1}}}$ and $\raisebox{.5pt}{\textcircled{\raisebox{-.9pt} {2}}}$ tokens cannot disappear in the process.

The main idea of Lemma~\ref{lem:cflucent} is to fire transitions that are enabled by 
agreement tokens until this is no longer possible. This leads to markings like $M_1' = [p2,p3,p5]$ and $M_2' = [p2,p4,p5]$ in Figure~\ref{fig-locally-safe-conflict-pair}. In such markings, all enabled transitions have a mix of agreement and disagreement tokens in their input places. For example, $t2$ is enabled in $M_1'$ by a $\bullet$ token in $p2$ and a $\raisebox{.5pt}{\textcircled{\raisebox{-.9pt} {1}}}$ token in $p3$, and 
$t3$ is enabled in $M_2'$ by a $\bullet$ token in $p5$ and a $\raisebox{.5pt}{\textcircled{\raisebox{-.9pt} {2}}}$ token in $p4$.
Lemma~\ref{lem:cflucent} shows that it is always possible to reach such markings using the fact that the home marking $\mi{Mrk}(C)$ is always reachable. 
Later, we will show that free-choice nets having a home cluster cannot have conflict-pairs.
Therefore, Figure~\ref{fig-locally-safe-conflict-pair} need to use an example that 
does not have a home cluster.

The proof of Lemma~\ref{lem:cflucent} can be summarized as follows.
Start from two different markings $M_1$ and $M_2$ that enable the same set of transitions. 
The tokens of both markings are split into ``agreement tokens'' denoted by $\bullet$ and ``disagreement tokens'' marked by $\raisebox{.5pt}{\textcircled{\raisebox{-.9pt} {1}}}$ or $\raisebox{.5pt}{\textcircled{\raisebox{-.9pt} {2}}}$ (as shown in Figure~\ref{fig-locally-safe-conflict-pair}).
Next, we fire transitions that consume only $\bullet$ tokens.
It is possible to do this in such a way that the process stops
and there are no such transitions enabled anymore (just try to move tokens closer to the home marking, this must stop at some stage because the disagreement tokens are needed). 
The $\raisebox{.5pt}{\textcircled{\raisebox{-.9pt} {1}}}$ and $\raisebox{.5pt}{\textcircled{\raisebox{-.9pt} {2}}}$ tokens do not move
and enabled transitions require at least one ``disagreement token'' ($\bullet$).
This way we can construct a conflict-pair $(M_1',M_2')$.
Hence, if there are no conflict-pairs, there cannot be two markings $M_1$ and $M_2$ that enable the same set of transitions, thus proving lucency.

\begin{lemma}[Nets Without Conflict-Pairs Are Lucent]\label{lem:cflucent}
Let $(N,M)$ be a marked proper free-choice net having a home cluster.
If $(N,M)$ has no conflict-pairs, then $(N,M)$ is lucent.
\end{lemma}
\begin{proof}
Let $(N,M)$ be a marked proper free-choice net having a home cluster $C$. We need to prove that if $N=(P,T,F)$ has no conflict-pairs, then $N$ is lucent. This can be rewritten to the logically equivalent contrapositive ``if $N$ is not lucent, then $N$ has a conflict-pair''. We will construct such a conflict-pair. 

Assume $N$ is \emph{not} lucent. There must be two markings $M_1,M_2 \in R(N,M)$ such that $\mi{en}(N,M_1)=\mi{en}(N,M_2)$ and $M_1 \neq M_2$.
We will  show that, based on these markings, we can construct a conflict-pair $(M_1',M_2')$.

The only dead reachable marking is $\mi{Mrk}(C)$.
Since $\mi{en}(N,M_1)=\mi{en}(N,M_2)$ and $M_1 \neq M_2$, we conclude
that $\mi{en}(N,M_1)=\mi{en}(N,M_2) \neq \emptyset$, $M_1 \neq \mi{Mrk}(C)$, and $M_2 \neq \mi{Mrk}(C)$ (use see Proposition~\ref{prop:cluschar}).

Since $(N,M)$ is safe (see Corollary~\ref{corr:safe}), we can 
partition the tokens into three groups based on the places where they reside:
$P_{\bullet} = \{p \in P \mid  p \in M_1 \ \wedge \  p \in M_2  \}$,
$P_{1} = \{p \in P \mid  p \in M_1 \ \wedge \  p \not\in M_2  \}$, and
$P_{2} = \{p \in P \mid  p \not\in M_1 \ \wedge \  p \in M_2  \}$.
Tokens in $P_{\bullet}$ are shared by both markings 
(i.e., the ``agreement tokens'' mentioned before).
Tokens in $P_{1}$ and $P_{2}$ exist in only one of the two markings (i.e., the ``disagreement tokens'' mentioned before).
None of these three sets can be empty.
Because $\mi{en}(N,M_1)=\mi{en}(N,M_2) \neq \emptyset$, transitions enabled in both markings must agree on the marked input places. Hence, $P_{\bullet} \neq \emptyset$.
Because $M_1 \neq M_2$ and one cannot be strictly larger than the other one (Corollary~\ref{corr:dommark}), $P_{1} \neq \emptyset$ and $P_{2} \neq \emptyset$.
We also create three groups of transitions:
$T_1 = \{t \in T \mid \pre{t} \cap P_{1} \neq \emptyset \}$,
$T_2 = \{t \in T \mid \pre{t} \cap P_{2} \neq \emptyset \}$,  and
$T_{\mi{rest}} = T \setminus (T_1 \cup T_2) = \{t \in T \mid \pre{t} \cap (P_{1} \cup P_{2}) = \emptyset \}$.
Note that $T_1$ and $T_2$ may overlap in principle, but do not overlap with $T_{\mi{rest}}$ i.e.,
$(T_1 \cup T_2)$ and $T_{\mi{rest}}$ partition $T$.
Each subset (i.e., $T_1$, $T_2$ or $T$) includes all transitions of a cluster or none (i.e., clusters agree on membership).

After introducing these notations, we pick a firing sequence 
$\sigma$ starting in $M_1$ and ending in the home marking, i.e.,
$(N,M_1)\allowbreak[\sigma\rangle\allowbreak (N,\mi{Mrk}(C))$.
Such a $\sigma$ exists, because $C$ is a home cluster.

Like in the proof of Theorem~\ref{theo:dommark} we split $\sigma$ into $\sigma_1$ and $\sigma_2$.
$\mi{Exp}_{(N,M_1)}(\sigma)$ contains all permutations of firing sequence $\sigma$ that are obtained by expediting transitions.
Let $\sigma_1$ and $\sigma_2$ be such that $\sigma_1 \cdot \sigma_2 \in \mi{Exp}_{(N,M_1)}(\sigma)$, 
$(N,M_1)[\sigma_1\rangle (N,M_1')$,
$\sigma_1 \in {T_{\mi{rest}}}^*$,
and $\mi{en}(N,M_1') \cap T_{\mi{rest}} \cap \{t \in \sigma_2\} = \emptyset$.
In other words, we expedite transitions from $T_{\mi{rest}}$, until this is no longer possible.
Given $\mi{Exp}_{(N,M_1)}(\sigma)$ it is always possible to find such $\sigma_1$, $\sigma_2$, and $M_1'$.
Suppose that $\mi{en}(N,M_1') \cap T_{\mi{rest}} \cap \{t \in \sigma_2\} \neq \emptyset$,
then we take the first transition in $\sigma_2$ that is in this set and move it to $\sigma_1$
(see construction in Definition~\ref{def:expedite}).
Since $\sigma_1$ does not fire transitions possibly consuming disagreement tokens 
(recall that $T_{\mi{rest}} = \{t \in T \mid \pre{t} \cap (P_{1} \cup P_{2}) = \emptyset \}$), $\sigma_1$ is also enabled in $M_2$.
Let $M_2'$ be the marking reached after firing $\sigma_1$ in $M_2$, i.e., $(N,M_2)[\sigma_1\rangle(N,M_2')$.
Also, $(N,M_1')[\sigma_2\rangle\allowbreak (N,\mi{Mrk}(C))$ (because $\sigma_1 \cdot \sigma_2 \in \mi{Exp}_{(N,M_1)}(\sigma)$ and Lemma~\ref{lemma:reord}).
Figure~\ref{fig-proof-cfp} summarizes the different entities involved and their relationships.
\begin{figure}[thb!]
\centering
\includegraphics[width=6.0cm]{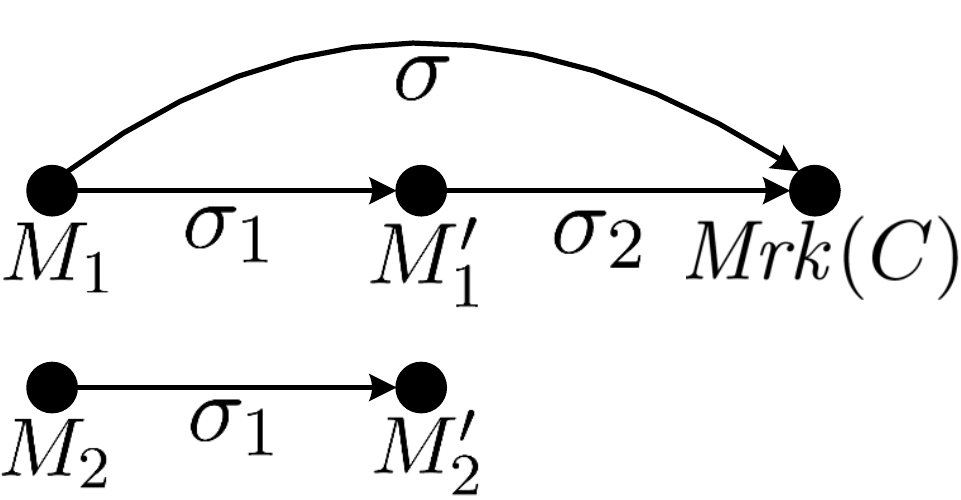}
\caption{Sketch of the construction used in Lemma~\ref{lem:cflucent}.
The elements satisfy the following relations:
$(N,M_1)\allowbreak[\sigma\rangle\allowbreak (N,\mi{Mrk}(C))$,
$\sigma_1 \cdot \sigma_2 \in \mi{Exp}_{(N,M_1)}(\sigma)$,
$(N,M_1)[\sigma_1\rangle (N,M_1')$,
$(N,M_2)[\sigma_1\rangle(N,M_2')$,
$(N,M_1')[\sigma_2\rangle (N,\mi{Mrk}(C))$,
$\sigma_1 \in {T_{\mi{rest}}}^*$, and
$\mi{en}(N,M_1') \cap T_{\mi{rest}} \cap \{t \in \sigma_2\} = \emptyset$.}
\label{fig-proof-cfp}
\end{figure}

$M_1(p) = M_1'(p)$ and $M_2(p) = M_2'(p)$ for any $p \in P_1 \cup P_2$, i.e., 
the disagreement places are unaffected by $\sigma_1$ because the $T_1$ and $T_2$ transitions did not fire
and $\sigma_1$ cannot add tokens to $P_{1}$ or $P_{2}$, because the net is safe 
(see Corollary~\ref{corr:safe}). $\sigma_1$ only produces ``agreement tokens'' and
putting such a token in a disagreement place violates the safety property in the sequence starting in $M_1$ or $M_2$.
Also $M_1(p) = M_2(p)$ and $M_1'(p) = M_2'(p)$ for any $p \in P \setminus (P_1 \cup P_2)$.
This also holds for intermediate markings when firing the transitions in $\sigma_1$. Hence, the collection of $\raisebox{.5pt}{\textcircled{\raisebox{-.9pt} {1}}}$ and $\raisebox{.5pt}{\textcircled{\raisebox{-.9pt} {2}}}$ tokens does not change (no disagreement tokens are removed and no new disagreement tokens are created). Moreover, there is a non-empty set of agreement tokens (denoted by $\bullet$) because the net is proper ($M_1'$ and $M_2'$ agree on these and each transition in $\sigma_1$ produces at least one such token).

Next, we prove that $(M_1',M_2')$ is indeed a conflict-pair for $N$.
We check the required properties listed in Definition~\ref{def:cfpair}:
\begin{itemize}[noitemsep,topsep=2pt]

\item $M_1'$ and $M_2'$ are indeed reachable markings of $(N,M)$ because $M_1,M_2 \in R(N,M)$, 
$(N,M_1)\allowbreak [\sigma_1\rangle(N,M_1')$ and $(N,M_2)[\sigma_1\rangle(N,M_2')$,
\item $M_1'$ contains at least one ``disagreement token'' $\raisebox{.5pt}{\textcircled{\raisebox{-.9pt} {1}}}$ and one ``agreement token'' $\bullet$ (see above).
$M_1'$ cannot be dead, because the only reachable marking that may be dead is $\mi{Mrk}(C)$ having a single token (apply again Proposition~\ref{prop:cluschar}).
$M_2'$ also contains at least one ``disagreement token'' $\raisebox{.5pt}{\textcircled{\raisebox{-.9pt} {2}}}$ and one ``agreement token'' $\bullet$.
Hence, neither $M_1'$ nor $M_2'$ can be dead.

\item Next, we show that $T' = \mi{en}(N,M_1') \cap \mi{en}(N,M_2') = \emptyset$ using the following observations:
\begin{itemize}[noitemsep,topsep=2pt]
\item $T' \subseteq T_{\mi{rest}}$, because the transitions in $T_1$ and $T_2$ cannot be enabled in both $M_1'$ and $M_2'$
(no tokens were added to a place in $P_1 \cup P_2$ by $\sigma_1$).
\item $\mi{en}(N,M_1') \cap T_{\mi{rest}} \cap \{t \in \sigma_2\} = \emptyset$ was used as a  criterion when splitting $\sigma$ into $\sigma_1$ and $\sigma_2$.
\item Combining the above
implies $T' \cap \{t \in \sigma_2\} = \emptyset$. Hence,
the input places of the transitions in $T'$ are still marked after executing $\sigma_2$ in $M_1'$.
\item Since $(N,M_1')[\sigma_2\rangle (N,\mi{Mrk}(C))$, the input places of $T'$ must be marked in $\mi{Mrk}(C)$. Hence, $T' \subseteq C$.
\item This implies that the transitions in the home cluster are enabled in both $M_1'$ and $M_2'$.
This is only possible if $M_1'=M_2'$ leading to a contradiction, i.e., $\mi{en}(N,M_1') \cap \mi{en}(N,M_2') = \emptyset$.
\end{itemize}
Note that in $M_1'$ and $M_2'$ all enabled transitions need to consume at least one ``disagreement token''. Hence, no transition can be enabled in both $M_1'$ and $M_2'$.
If a transition would be enabled in both, then $\sigma_1$ could have been extended.

\item For all $t \in \mi{en}(N,M_1')$: $M_2'(\pre t)\geq 1$, because each transition enabled in $M_1'$ must 
have an ``agreement token'' produced by $\sigma_1$ and a ``disagreement token'' in $P_1$.
If a transition $t$ would be enabled based on ``disagreement tokens'' only, these would have been there in $M_1$ already
(recall that $M_1(p) = M_1'(p)$ for any $p \in P_1$) leading to a contradiction because $\mi{en}(N,M_1)=\mi{en}(N,M_2)$.
Hence, any transition $t$ enabled in $M_1'$ must have an ``agreement token'' produced by $\sigma_1$ on one of it input places.
This token is also there in $M_2'$. Hence, $M_2'(\pre t)\geq 1$.

\item For all $t \in \mi{en}(N,M_2')$: $M_1'(\pre t)\geq 1$. Here the same arguments apply.
A transition cannot be enabled based on ``disagreement tokens'' only, since these would have been there in $M_2$ already
($M_2(p) = M_2'(p)$ for any $p \in P_2$).
\end{itemize}
Hence, $(M_1',M_2')$ is indeed a conflict-pair and thus completes the contrapositive proof.
\end{proof}

\subsection{Home Clusters Ensure Lucency in Free-Choice Nets}
\label{sec:mainresult}

Now we can prove the main result of this paper: Marked proper free-choice nets having a home cluster are lucent.
We use the notions of rooted disentangled paths and
conflict-pairs. The basic idea is to show that 
a conflict-pair implies that there is an unsafe rooted disentangled path which is impossible.
The absence of conflict-pairs implies lucency.
\begin{figure}[thb!]
\centering
\includegraphics[width=16.0cm]{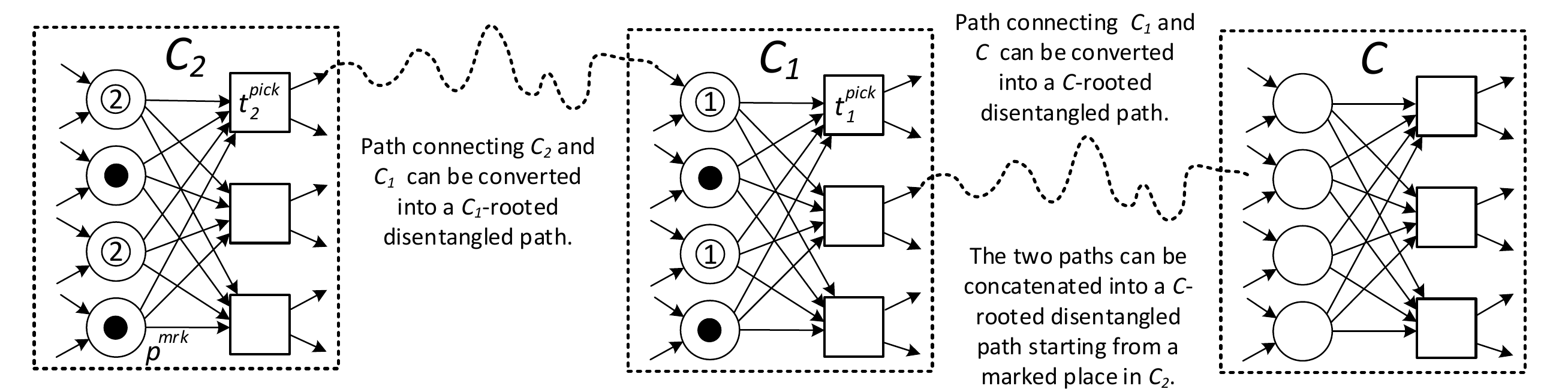}
\caption{Visualization of the three clusters considered in the proof of Theorem~\ref{theo:home-cp}. $C$ is the home cluster. $C_1$ is a cluster enabled in $M_1$ but not in $M_2$.
$C_2$ is a cluster enabled in $M_2$ but not in $M_1$.
The places labeled $\raisebox{.5pt}{\textcircled{\raisebox{-.9pt} {1}}}$ or $\raisebox{.5pt}{\textcircled{\raisebox{-.9pt} {2}}}$ contain a token in the respective marking (just in $M_1$ or just in $M_2$). The paths connecting $C_2$ and $C_1$ and $C_1$ and $C$ are converted into rooted disentangled paths.
These two rooted disentangled paths can be concatenated to create a $C$-rooted disentangled path starting in $p^{mrk}$. The proof shows that at least one of these rooted disentangled paths is not safe, proving that the net has no conflict-pairs and thus must be lucent.}
\label{fig-theo-lucent}
\end{figure}

Theorem~\ref{theo:home-cp} shows that there cannot be a conflict-pair $(M_1,M_2)$ in
a marked proper free-choice net having a home cluster.
Figure~\ref{fig-theo-lucent} sketches the main idea of the proof.
First, we assume that there exists a conflict-pair $(M_1,M_2)$.
We identify, next to the home cluster $C$, two additional clusters $C_1$ and $C_2$ based on 
the conflict-pair $(M_1,M_2)$. 
$C_1$ is enabled in marking $M_1$ and $C_2$ is enabled in marking $M_2$.
$C_1$ can be any cluster enabled in marking $M_1$.
$C_2$ is a cluster enabled in marking $M_2$ that contributes to the enabling of
cluster $C_1$ which is disabled in marking $M_2$.
As Figure~\ref{fig-theo-lucent} shows $C_1$ has a $\raisebox{.5pt}{\textcircled{\raisebox{-.9pt} {1}}}$ input token and $C_2$ has a $\raisebox{.5pt}{\textcircled{\raisebox{-.9pt} {2}}}$ input token.

Based on the selected  $C_1$ and $C_2$ clusters, we create two rooted disentangled paths: 
$\rho'$ is a $C_1$-rooted disentangled path connecting $C_2$ to $C_1$ and 
$\rho''$ is a $C$-rooted disentangled path connecting $C_1$ to $C$.
These two rooted disentangled paths are combined into a path $\rho'''$ running from $C_2$ to $C$ via $C_1$.
If $\rho'''$ is not a $C$-rooted disentangled path (i.e., the same cluster is visited multiple times along the path), then it is possible to reach a marking starting from $M_2$ which puts a token on $\rho''$
(the path connecting $C_1$ to $C$) while having an agreement token in $C_1$.
Hence, there is a $C$-rooted disentangled path connecting $C_1$ to $C$ having at least two tokens (see proof for details). Using Lemma~\ref{lem:pathsafety} this leads to a contradiction. Hence, $\rho'''$ must be a $C$-rooted disentangled path. However, considering $M_1$ (rather than a marking reached from $M_2$) path $\rho'''$ has at least two tokens.
This also leads to a contradiction using Lemma~\ref{lem:pathsafety}. 
Therefore, there cannot be a conflict-pair $(M_1,M_2)$.
The approach presented using Figure~\ref{fig-theo-lucent} is detailed in the proof below. 

\begin{theorem}[Home Clusters Ensure Absence of Conflict-Pairs]\label{theo:home-cp}
Let $(N,M)$ be a marked proper free-choice net having a home cluster. $(N,M)$ has no conflict-pairs.
\end{theorem}
\begin{proof}
Let $(N,M)$ be a marked proper free-choice net having a home cluster $C$. 
We assume that $(N,M)$ has a conflict-pair $(M_1,M_2)$
and show that this leads to a contradiction.

{\bf Useful notations: $P_{\bullet}$, $P_{\emptyset}$, $P_{1}$, and $P_{2}$.}
Based on the conflict-pair $(M_1,M_2)$, we partition the set of places $P$ into four sets
$P_{\bullet} = \{p \in P \mid  p \in M_1 \ \wedge \  p \in M_2  \}$,
$P_{\emptyset} = \{p \in P \mid  p \not\in M_1 \ \wedge \  p \not\in M_2  \}$,
$P_{1} = \{p \in P \mid  p \in M_1 \ \wedge \  p \not\in M_2  \}$, and
$P_{2} = \{p \in P \mid  p \not\in M_1 \ \wedge \  p \in M_2  \}$.
Transitions enabled in $M_1$ have 
input places from $P_{\bullet}$ and $P_{1}$.
$\mi{en}(N,M_1) = \{ t \in T \mid 
\pre{t} \cap P_{\bullet} \neq \emptyset \ \wedge \
\pre{t} \cap P_{\emptyset} = \emptyset \ \wedge \
\pre{t} \cap P_{1} \neq \emptyset \ \wedge \
\pre{t} \cap P_{2} = \emptyset 
\}$.
Transitions enabled in $M_2$ have 
input places from $P_{\bullet}$ and $P_{2}$.
$\mi{en}(N,M_2) = \{ t \in T \mid 
\pre{t} \cap P_{\bullet} \neq \emptyset \ \wedge \
\pre{t} \cap P_{\emptyset} = \emptyset \ \wedge \
\pre{t} \cap P_{1}  = \emptyset \ \wedge \
\pre{t} \cap P_{2} \neq \emptyset 
\}$.
This follows directly from Definition~\ref{def:cfpair}.

{\bf Selecting clusters $C_1$ and $C_2$.}
Pick an arbitrary transition enabled in $M_1$:
$t^{pick}_1 \in \mi{en}(N,M_1)$. 
Call the corresponding cluster $C_1$ (i.e., $t^{pick}_1 \in C_1$).
Cluster $C_1$ is fully marked in $M_1$, 
but has at least one unmarked place in $M_2$.
$P^{unmrk} = C_1 \cap P_{1}$ is the non-empty set of such places.
To reach the home marking from $M_2$, we need to execute a transition in cluster $C_1$ because it is partially enabled.
Hence, there needs to be a firing sequence that marks the places in $P^{unmrk}$. Let $\sigma_{en}$ be a shortest firing sequence
starting in $M_2$ and marking a place in $P^{unmrk}$.
$\sigma_{en}$ starts with a transition enabled in $M_2$ and ends with a transition putting the first token in $P^{unmrk}$ (the transition may also mark other places in $P^{unmrk}$). 
Let $t^{pick}_2 \in \mi{en}(N,M_2)$ be the first transition in this shortest sequence $\sigma_{en}$.
Given this firing sequence we can ``follow a token'' 
from $t^{pick}_2$ to a place in $P^{unmrk}$.
This provides a path starting in $t^{pick}_2$ and ending in the first place marked in $P^{unmrk}$.
This path contains a subset of transitions in $\sigma_{en}$.
Obviously, such a path must exist, but there may be many candidates.
The cluster where this path starts is called $C_2$ (i.e., $t^{pick}_2 \in C_2$).
There exists a place $p^{mrk} \in C_2 \cap P_{\bullet}$ in this cluster that is marked in both $M_1$ and $M_2$
($t^{pick}_2$ is enabled in $M_2$ and at least of the input places must also have a token in $M_1$, since $(M_1,M_2)$ is a conflict pair).

{\bf Selecting two rooted disentangled paths $\rho'$ and $\rho''$.}
We use the three clusters $C$, $C_1$, and $C_2$ to prove the contradiction. There is a path from $C_2$ to $C_1$ and a path from $C_1$ to $C$. Note that $C_1$ and $C_2$ need to be different due to the disagreement tokens.
Also $C$ is different from both $C_1$ and $C_2$, since it is not possible to mark the home cluster completely and still have tokens in other places (use Corollary~\ref{corr:dommark}).
Due to Lemma~\ref{lem:exstrdespath}
there must be a $C_1$-rooted disentangled path starting in $p^{mrk}$. Let us call this path
$\rho' = \langle p_1',t_1',p_2', \ldots , t_{n-1}',p_n' \rangle$.
$p_1' = p^{mrk}$ and $p_n'$ is a place in cluster $C_1$.
Assume that the construction described in Lemma~\ref{lem:exstrdespath} is used, i.e., 
all transitions in $\rho'$ also appear in $\sigma_{en}$ (but the reverse does not need to hold since we follow a token and take shortcuts to ensure that each cluster appears only once).
For clarity, we refer to the end place of $\rho'$ as $p^{conn}$, i.e., 
$p^{conn} = p_n'$.
Due to Corollary~\ref{corr:pathexists}
there must also be a $C$-rooted disentangled path starting in $p^{conn}$ ($p^{conn}$ is non-dead in $(N,M)$).
Let us call this path
$\rho'' = \langle p_1'',t_1'',p_2'', \ldots , t_{m-1}'',p_m'' \rangle$.
$p_1'' = p^{conn}$ and $p_m''$ is a place in cluster $C$.
For clarity, we refer to this place as $p^{end}$, i.e., 
$p^{end} = p_m''$.

Hence, we have 
a $C_1$-rooted disentangled path $\rho'$ starting in $p^{mrk}$ and ending in $p^{conn}$ and
a $C$-rooted disentangled path $\rho''$ starting in $p^{conn}$ and ending in $p^{end}$.

{\bf Creating another rooted disentangled path $\rho'''$ by combining $\rho'$ and $\rho''$.}
Consider now the 
path $\rho''' = 
\langle p^{mrk},t_1',p_2', \ldots , t_{n-1}',p^{conn},
t_1'',p_2'', \ldots , t_{m-1}'',p^{end}
 \rangle$, i.e., the concatenation of the paths $\rho'$ and $\rho''$.
We will show that $\rho'''$ is a $C$-rooted disentangled path  starting in $p^{mrk}$ and ending in $p^{end}$.

Obviously, $\rho'''$ is also a path of $N$.
However, we also need to show that $\rho'''$
does not contain elements that belong to the same cluster.
If this is not the case there must be a place $p_i'$ in $\rho'$ with $1 \leq i < n$ and a place $p_j''$ in $\rho''$ with 
$1 \leq j \leq m$ that belong to the same cluster. (Note that $\rho'$  and $\rho''$ do not visit the same cluster twice when considered separately, and  $p_n'= p^{conn} = p_1''$ is in both so should not be compared with itself.)
However, this is impossible.
Assume there would be a cluster $C'$ 
with $p_i' \in C'$ and $p_j'' \in C'$.
Then a transition of this cluster should appear in $\sigma_{en}$.
Recall that we assume that the construction described in Lemma~\ref{lem:exstrdespath} is used to create $\rho'$, i.e., all transitions in $\rho'$ also appear in $\sigma_{en}$.
$t' \in C'$ is such a transition appearing in $\sigma_{en}$
and $\rho'$ and consuming tokens from both $p_i'$ and $p_j''$.
When starting in marking $M_2$ and executing $\sigma_{en}$, transition $t'$ occurs before any transition in $C_1$.
Consider the marking $M'$ just before $t'$ occurs, i.e., starting in $M_2$ a prefix of $\sigma_{en}$ is executed enabling $t'$ without executing any transition in $C_1$.
There exists a place $p^{alt} \in C_1 \cap P_{\bullet}$, because $C_1$ is fully marked in $M_1$ and partially marked in $M_2$.
In marking $M'$, both $p_j''$ and $p^{alt}$ are marked. 
$p_j''$ is marked because $t'$ is enabled.
$p^{alt}$ is marked because no transition in $C_1$ fired yet.
However, there is also a $C$-rooted disentangled path  
starting in $p^{alt}$, namely
$\rho^{alt} = \langle p^{alt},t_1'',p_2'', \ldots ,p_j'', \ldots t_{m-1}'',p^{end} \rangle$ 
(we can start in an arbitrary place in $C_1$ and still meet all requirements, note that compared to $\rho''$, $p^{conn}$ is replaced by $p^{alt}$).
Lemma~\ref{lem:pathsafety} shows that it is impossible 
to have two marked places $p_j''$ and $p^{alt}$ in the 
$C$-rooted disentangled path $\rho^{alt}$, leading to a contradiction.
Therefore, $\rho'''$ does not visit the same cluster multiple times (if so, $\rho''$ would not be a $C$-rooted disentangled path).
Hence, $\rho'''$ is a $C$-rooted disentangled path
starting in $p^{mrk}$ and ending in $p^{end}$.

{\bf The combined rooted disentangled path $\rho'''$ is not safe leading to a contradiction.}
Now consider the just constructed $C$-rooted disentangled path $\rho'''$ and marking $M_1$.
The places $p^{mrk}$ and $p^{conn}$ are both marked in $M_1$ and must be different.
Recall that $p^{mrk} \in C_2 \cap P_{\bullet}$ (i.e., also marked in $M_1$)
and $p^{conn} \in C_1$ (all places in $C_1$ are marked in $M_1$).
Again we apply Lemma~\ref{lem:pathsafety}, which shows that it is impossible to have two marked places 
in the $C$-rooted disentangled path $\rho'''$.
Therefore, we find another contradiction, showing that 
the conflict-pair $(M_1,M_2)$ cannot exist.
\end{proof}

Our goal was to show that marked proper free-choice nets having a home cluster
are lucent and this follows directly from the previous results.

\begin{corollary}[Home Clusters Ensure Lucency]\label{corr:home-lu}
Let $(N,M)$ be a marked proper free-choice net having a home cluster. $(N,M)$ is lucent.
\end{corollary}
\begin{proof}
This follows directly from Lemma~\ref{lem:cflucent} and Theorem~\ref{theo:home-cp}.
A marked proper free-choice net having a home cluster 
has no conflict-pairs (Theorem~\ref{theo:home-cp})
and, therefore, must be lucent (Lemma~\ref{lem:cflucent}).
\end{proof}

$(N_1,M_1)$ in Figure~\ref{fig-intro-fc} and 
$(N_5,M_5)$ in Figure~\ref{fig-home-lucent} are examples of 
free-choice nets having a home cluster and these are indeed lucent.
$(N_4,M_4)$ depicted in Figure~\ref{fig-fc-nonlucid} is
not lucent and indeed has no home cluster.

\section{Relation To Perpetual Nets}
\label{sec:perpmarknets}

This paper significantly extends the results for \emph{perpetual} marked 
free-choice nets presented in \cite{lucent-PN2018}.
These nets need to be live, bounded, and have a home cluster, 
whereas in this paper, we only require the latter (but boundedness is implied).
Moreover, unlike \cite{lucent-PN2018} the setting is not limited to strongly-connected nets, 
e.g., we allow for workflow nets and other types of Petri nets typically used in 
process mining, workflow management, and business process management.

\begin{definition}[Perpetual Marked Nets \cite{lucent-PN2018}]\label{def:perpmn}
A marked Petri net $(N,M)$ is perpetual if and only if 
it is live, bounded, and has a home cluster.
\end{definition}

In this paper, we focus on marked proper free-choice nets having a home cluster. 
Since boundedness is implied, the essential difference is the liveness requirement that we dropped.
None of the lucent Petri nets shown in this paper is live,
showing that this is a substantial generalization.
For example, $(N_1,M_1)$ in Figure~\ref{fig-intro-fc} and 
$(N_5,M_5)$ in Figure~\ref{fig-home-lucent}
are lucent but not perpetual. 
Lemma~\ref{lem:cflucent} and Theorem~\ref{theo:home-cp} (combined in Corollary~\ref{corr:home-lu}) can be used to show that 
$(N_1,M_1)$ and $(N_5,M_5)$ are lucent.
\begin{table}[htb!]
\centering
\begin{tabular}{|p{1.5cm}|p{3.0cm}||c|c|}
  \hline
\multicolumn{2}{|p{4.7cm}||}{class of nets for which lucency is proven to hold} & \parbox[t]{4cm}{marked proper free-choice nets having a home cluster (this paper)} & 
\parbox[t]{4cm}{perpetual nets (free-choice, live, bounded, and having home cluster) \cite{lucent-PN2018}} \\
    \hline \hline
structural & proper & \checkmark & \checkmark ~ (implied)\\ \cline{2-4}
properties & strongly-connected & - & \checkmark  ~ (implied)\\
    \hline \hline
dynamic & bounded & \checkmark ~ (implied) & \checkmark\\ \cline{2-4}
properties & live & - & \checkmark\\ \hline
\end{tabular}
\caption{Corollary~\ref{corr:home-lu} extends the results in \cite{lucent-PN2018} to nets that may be non-live and not strongly-connected (requirements are denoted by $\checkmark$).}\label{tab:prop}
\end{table}

Theorem~3 in \cite{lucent-PN2018} states that any perpetual marked free-choice net is lucent.
Corollary~\ref{corr:home-lu} generalizes this statement, as shown in Table~\ref{tab:prop}.
In the remainder of this section, we relate both settings.

\begin{proposition}[Perpetual Nets Are a Subclass of Free-Choice Nets Having a Home Cluster]\label{lemma:implication}
Let $(N,M)$ be a marked free-choice net.
If $(N,M)$ is perpetual, then $(N,M)$ is proper and has a home cluster.
\end{proposition}
\begin{proof}
A marked free-choice net $(N,M)$ is a perpetual net if and only if 
it is live, bounded, and has a home cluster.
Hence, we only need to show that $(N,M)$ is proper.
This follows directly from the fact that well-formed nets are strongly-connected (Theorem 2.25 in \cite{deselesparza}).
\end{proof}

The reverse does not need to hold, as is demonstrated by figures~\ref{fig-intro-fc} and 
\ref{fig-home-lucent}.
The proof of Theorem~3 in \cite{lucent-PN2018} is also incomplete.
The proof in \cite{lucent-PN2018} can be repaired, but this requires reasoning over a stacked array of P-components, making things overly complicated.
It is also possible to use a different approach using a so-called T-reduction showing the absence of conflict pairs, 
see Theorem~6 in \cite{reduction-wvda-PN2021}.
In a T-reduction proper $t$-induced T-nets are ``peeled off'' until a T-net (i.e., marked graph) 
remains (this is related to the notion of CP-nets used in \cite{deselesparza}).
The reduction preserves liveness, boundedness, perpetuality, pc-safeness, and other properties.
Starting from a perpetual well-formed free-choice net and a T-reduction, it can be shown that
lucency is preserved in the ``upstream'' direction.
Since for marked graphs it is easy to show lucency, 
this implies that any perpetual marked free-choice net is lucent.

Selected results from Section~\ref{sec:arelucent} can also be used to repair the proof in \cite{lucent-PN2018} in a more direct manner 
without using existing results for well-formed free-choice nets. 
In this more limited setting, our approach can be further simplified by exploiting safeness and liveness.

For strongly-connected marked free-choice nets, having a home cluster implies perpetuality (i.e., liveness and boundedness are implied).
Moreover, such nets are also safe.

\begin{proposition}[Properties of Strongly-Connected Free-Choice Nets Having a Home Cluster]\label{prop:live-safe}
A strongly-connected marked free-choice net $(N,M)$ 
having a home cluster $C$ is live, safe, and lucent.
\end{proposition}
\begin{proof}
Let $(N,M)$ be a strongly-connected marked free-choice net having a home cluster $C$.
$N$ is proper because $N$ is strongly-connected.
Hence, we can apply Corollary~\ref{corr:safe} to show that $(N,M)$ is safe.
Corollary~\ref{corr:home-lu} can be used to show that $(N,M)$ is lucent.
Any transition $t$ is on a path from starting in $C$. 
It is possible to create a firing sequence starting in $\mi{Mrk}(C)$ enabling $t$ by following this path. 
This is due to the free-choice property and the fact that we cannot ``get stuck on the way'' (it is always possible to return to $\mi{Mrk}(C)$).
See the proof of Lemma~\ref{lem:pathsafety} for a similar reasoning.
Hence, $(N,M)$ is live.
\end{proof}

To explore the relationship between both settings in more detail, 
we take a proper Petri net with a safe initial marking $M$ and a selected cluster $C$.
We add a transition $t_C$ that extends the cluster and that marks all places in $M$,
i.e., $\pre{t_C} = C \cap P$ and $\post{t_C} = \{p \in M\}$.
$t_C$ short-circuits the original net in an attempt to make it strongly-connected.
To achieve this, we also need to remove the nodes for which there is no path from the initially marked places.

\begin{definition}[Short-Circuited Cleaned Nets]\label{def:short-circ}
Let $N= (P,T,F)$ be proper Petri net having a cluster $C$ and an initial marking $M$ that is safe.
\begin{itemize}[noitemsep,topsep=2pt]
\item $\mi{conn}(N,M) = \{ x_n \mid \langle x_1,x_2, \ldots ,x_n \rangle \in \mi{paths}(N) \ \wedge \ x_1 \in M \}$ are all nodes that are on a path starting in an initially marked place.
\item $\mi{clean}(N,M) = (P',T',F'\cap ((P' \times T') \cup (T' \times P')))$ with
$P' = P \cap \mi{conn}(N,M)$, and $T' = T \cap \mi{conn}(N,M)$ is the net containing all places and transitions on paths starting in an initially marked place.
\item $\mi{short\_circ}(N,C,M) = (P,T\cup \{t_C\},F \cup (\mi{Pl}(C) \times \{t_C\}) 
\cup (\{t_C\}\times \{p \in M\}) )$ is the short-circuited net (adding a ``fresh'' transition $t_C \not\in T$
with $\pre{t_C} = C \cap P$ and $\post{t_C} = \{p \in M\}$).
\item $N_{C,M} = \mi{short\_circ}(\mi{clean}(N,M),\allowbreak C,M)$ applies the two operations in sequence.
\item $\hat{C} = C \cap \{t_c\}$ is used to denote the extended cluster (note that this is only a cluster of $N_{C,M}$ if $C \subseteq \mi{conn}(N,M)$).
\end{itemize}
\end{definition}

In an attempt to create a strongly-connected net, we first remove all ``dead nodes'' and then short-circuit the net by connecting a selected cluster to the initially marked places.
If all nodes of $C$ are on a path starting in an initially marked place, then 
$\hat{C} = C \cap \{t_c\}$ is indeed a cluster of $N_{C,M}$ (otherwise not).

\begin{proposition}[Short-Circuited Cleaned Nets Are Strongly-Connected]\label{prop:scc-nets-are-cc}
Let $(N,M)$ be a safely marked proper free-choice net having a cluster $C$ such that $C \subseteq \mi{conn}(N,M)$.
The short-circuited cleaned net $N_{C,M} = \mi{short\_circ}(\mi{clean}(N,M),\allowbreak C,M)$ is
strongly-connected and free-choice, and
$\hat{C} = C \cap \{t_c\} \in \cluster{N_{C,M}}$
(i.e., $\hat{C}$ is indeed a cluster of $N_{C,M}$).
\end{proposition}
\begin{proof}
All nodes in $\mi{clean}(N,M)$ are reachable from an initially marked place
(including the nodes in $C$ because $C \subseteq \mi{conn}(N,M)$).
Hence, $t_c$ is also reachable from an initially marked place 
and $t_c$ is connected to this place. Therefore, the net is strongly-connected. 
Adding $t_c$ cannot destroy the free-choice property. 
If there is a transition $t \in C$, then $\pre{t}=\pre{t_c}$.
If not, then $C$ has just one place.
Therefore, $N_{C,M}$ is free-choice and has a new cluster $\hat{C} = C \cap \{t_c\}$.
\end{proof}

Under the assumption that cluster $C$ is preserved when short-circuiting the net,  $C$ is a home cluster of $(N,M)$ if and only if $\hat{C}$ is a home cluster of $(N_{C,M},M)$.
Moreover, this is equivalent to $(N_{C,M},M)$ being live and bounded,
and can be used to decide whether a free-choice net has a home cluster in polynomial time.

\begin{theorem}[Relating Both Settings]\label{theo:relation}
Let $(N,M)$ be a safely marked proper free-choice net having a cluster $C$ such that $C \subseteq \mi{conn}(N,M)$. The following three statements are equivalent:
\begin{itemize}[noitemsep,topsep=2pt]
\item[(1)] $C$ is a home cluster of $(N,M)$,
\item[(2)] $\hat{C}$ is a home cluster of $(N_{C,M},M)$, and
\item[(3)] $(N_{C,M},M)$ is live and bounded.
\end{itemize}
\end{theorem}
\begin{proof}
Let $N= (P,T,F)$ be a proper free-choice net having a cluster $C$ 
and an initial marking $M$ that is safe.
$N_{C,M} = \mi{short\_circ}(\mi{clean}(N,M),\allowbreak C,M)$ and $t_C$ is the short-circuiting transition. 
$C \subseteq \mi{conn}(N,M)$, i.e., all cluster nodes are reachable from an initially marked place.

First, we show that (1) $\Rightarrow$ (2). 
Assume that $C$ is a home cluster of $(N,M)$.
Under this assumption, we consider the reachable markings of $(N_{C,M},M)$.
These include the markings of $(N,M)$, but nothing more.
The moment all places in $C$ are marked, the other places are empty.
In $(N_{C,M},M)$ there is an additional transition $t_C$ 
that is enabled if all places in $\hat{C}$ are enabled. 
If $t_C$ fires in $\mi{Mrk}(C) = \mi{Mrk}(\hat{C})$, 
then we reach the initial state $M$ again.
Hence, the set of reachable markings is the same and 
$\hat{C}$ is a home cluster of $(N_{C,M},M)$.

Second, we show that (2) $\Rightarrow$ (3). 
$\hat{C}$ be a home cluster of $(N_{C,M},M)$.
Proposition~\ref{prop:scc-nets-are-cc} shows that 
$N_{C,M}$ is strongly-connected and free-choice.
Using Proposition~\ref{prop:live-safe} this implies that
$(N_{C,M},M)$ is live and safe (i.e., also bounded). 

Finally, we show that (3) $\Rightarrow$ (1).
Let $(N_{C,M},M)$ be live and bounded.
This implies that also $t_C$ is live and can be repeatedly be enabled.
When $t_C$ is enabled, the places in $\hat{C}$ are marked, i.e.,
$t_C$ can only be enabled in a marking $M'$ such that $M' \geq \mi{Mrk}(C)$.
It is impossible that $M' > \mi{Mrk}(C)$. If so, it would be possible to reach a marking larger than the initial marking yielding an unbounded net by firing $t_C$.
Hence, $M' = \mi{Mrk}(C)$ is the only reachable marking enabling $t_C$.
Therefore, the set of reachable markings of $(N_{C,M},M)$ and $(N,M)$ are the same. As a result, $\mi{Mrk}(C)$ can be reached from any reachable marking starting in $(N,M)$. This implies that $C$ is a home cluster of $(N,M)$.

Combining (1) $\Rightarrow$ (2), (2) $\Rightarrow$ (3), and (3) $\Rightarrow$ (1)
shows that the three statements are equivalent.
\end{proof}

We can apply Theorem~\ref{theo:relation} to all clusters of the net. 
Therefore, the problem of deciding whether marked proper free-choice net has a home cluster
can be converted into a liveness and boundedness question, allowing us to solve the problem in polynomial time.

\begin{corollary}[Complexity of Home Cluster Detection]\label{corr:complex}
The following problem is solvable in polynomial time:
Given a marked proper free-choice net, to decide whether there is a home cluster.
\end{corollary}
\begin{proof}
Let $(N,M)$ be a marked proper free-choice net with $N=(P,T,F)$.
There are at most $\card{P}$ clusters. For each cluster $C$, we check whether 
$C$ is a home cluster of $(N,M)$. This is the same as checking whether $C \subseteq \mi{conn}(N,M)$ and
$(N_{C,M},M)$ is live and bounded. 
The former requirement is merely a syntactical check to ensure that
cluster $C$ is preserved when short-circuiting the net.
The latter requirement is known to be solvable in polynomial time 
(see, for example, Corollary~6.18 in \cite{deselesparza}).
Hence, deciding whether there is a home cluster can also be solved in polynomial time.
\end{proof}

The above result is remarkable because it also applies to non-well-formed nets.

\section{Conclusion}
\label{sec:concl}

This paper shows that marked proper free-choice nets 
having a home cluster are lucent. 

A system is \emph{lucent} if the set of enabled actions
uniquely characterizes the state of the system.
The user interface of an information system or
the worklist provided by a workflow management system offers 
possible actions to its users.
If the system is not lucent, the system may behave differently in seemingly identical situations. 
The notion of lucency was introduced in \cite{lucent-PN2018}
and, given its foundational nature, it is surprising 
that this was not investigated before.

The paper focuses on \emph{marked proper free-choice nets 
having a home cluster} and uses novel concepts such as \emph{rooted disentangled paths} and \emph{conflict-pairs} to reason about 
the behavior of such models. Most of the work on free-choice nets 
is restricted to well-formed nets. However, the liveness requirement is unsuitable for many application domains.
Many systems and processes are \emph{terminating} and/or have an \emph{initialization} phase. These are excluded by most of the existing work. 
As shown in this paper, we can often short-circuit the net and apply existing techniques.
However, the approach used in this paper is direct without using any results for well-formed free-choice nets.

Future work aims to extend the class of systems for which lucency can be proven. However, this will not be easy since unbounded nets or nets with long-term dependencies are inherently non-lucent.
More promising is the further investigation of Petri nets with home clusters. Ideas such as rooted disentangled paths and conflict-pairs have a broader applicability and may be used to generalize some of the results known for well-formed (free-choice) Petri nets.
For example, is it possible to create reduction and synthesis rules?

The idea to look into lucency originated from challenges in the field of \emph{process mining} (where observed behavior without state information is converted into process models that have states).
What if event logs would not only show the actions executed, but also what was possible, but did not happen?
In \cite{lucent-translucent-fi-2019} the notion of translucent event logs is introduced, and a baseline discovery algorithm is given. Given such information, it is much easier to discover process models. Another direction for future research is to create process mining techniques tailored towards discovering a marked proper free-choice net having a home cluster from a standard event log. Current approaches often aim to discover workflow nets that are (relaxed) sound. 
Heuristic approaches do not ensure soundness.
Region-based techniques tend to create unreadable models.
Inductive mining techniques tend to produce underfitting models.
Therefore, there is room for exploring alternative representational biases in process mining.

~\\
{\bf Acknowledgements}: The author thanks the Alexander von Humboldt (AvH) Stiftung for supporting our research.
Special thanks go to the persistent anonymous reviewer for providing detailed comments that helped to improve the readability of the proofs.

\bibliographystyle{fundam}
\bibliography{lit}

\end{document}